\documentclass[journal]{IEEEtran}

\title{Autonomous Task Offloading of Vehicular Edge Computing with Parallel Computation Queues}
\author{Sungho Cho, Sung Il Choi, Seung Hyun Oh, Ian P. Roberts, 
and Sang Hyun Lee
\thanks{This work was supported the National Research Foundation of Korea (NRF) grant funded by the Korea government (MSIT) (2022R1A5A1027646 and RS-2025-00563388).
Sungho Cho and Ian P. Roberts are with Department of Electrical and Computer Engineering, University of California, Los Angeles, Los Angeles, CA, 90024, USA (e-mail: shcho1304@g.ucla.edu, ianroberts@ucla.edu).
Sung Il Choi, Seung Hyun Oh and Sang Hyun Lee are with the School of Electrical Engineering, Korea University, Seoul 02841, South Korea (e-mail: sungchoi@korea.ac.kr, seunghyunoh@korea.ac.kr, sanghyunlee@korea.ac.kr).}
}
\usepackage{amscd,graphics,dsfont,kotex}
\usepackage{cite}
\usepackage{graphicx}
\usepackage{amsthm}
\usepackage{amsmath}
\usepackage{amssymb}
\usepackage{amsfonts}
\usepackage{subfigure}
\usepackage{mathtools}
\usepackage{algorithm}
\usepackage{algorithmic}
\usepackage[dvipsnames]{xcolor}
\usepackage{color}
\usepackage{epstopdf}
\usepackage{caption}
\usepackage{verbatim}
\usepackage{float}
\usepackage{array}
\usepackage[english]{babel}
\usepackage[normalem]{ulem}
\usepackage[font=footnotesize,labelsep=period]{caption}
\usepackage{hhline}

\newtheorem{prop}{Proposition}
\newtheorem{lemma}{Lemma}

\newtheorem{definition}{Definition}
\newtheorem{theorem}{Theorem}

\newcommand{\floor}[1]{{\left\lfloor{#1}\right\rfloor}}

\newcommand{\rom}[1]{\lowercase\expandafter{\romannumeral #1\relax}}

\begin{document}
\maketitle

\begin{abstract}
This work considers a parallel task execution strategy in vehicular edge computing (VEC) networks, where edge servers are deployed along the roadside to process offloaded computational tasks of vehicular users. 
To minimize the overall waiting delay among vehicular users, a novel task offloading solution is implemented based on the network cooperation balancing resource under-utilization and load congestion. 
Dual evaluation through theoretical and numerical ways shows that the developed solution achieves a globally optimal delay reduction performance compared to existing methods, which is also validated by the feasibility test over a real-map virtual environment. 
The in-depth analysis reveals that predicting the instantaneous processing power of edge servers facilitates the identification of overloaded servers, which is critical for determining network delay.
By considering discrete variables of the queue, the proposed technique's precise estimation can effectively address these combinatorial challenges to achieve optimal performance.
\end{abstract}



\section{Introduction}
Recent advances in artificial intelligence and wireless communication have accelerated the practical deployment of intelligent transportation applications \cite{garcia2021}.
Consequently, new delay-sensitive vehicular services have emerged, such as autonomous navigation, localization, and driving.
These services typically rely on processing large volumes of real-time measurement data under highly dynamic and time-varying vehicular conditions.
Meanwhile, vehicles are inherently constrained by time-varying wireless links, limited data storage, and limited computational capabilities \cite{liu2021vehicular}.
These limitations pose significant challenges in meeting stringent delay requirements, particularly for computationally intensive tasks such as natural language processing and large language models.

To address these challenges, vehicular edge computing (VEC), which consists of edge servers for task processing and roadside units (RSUs) to communicate with vehicles, has been introduced to provide computational resources by virtue of the physical proximity of edge servers to vehicles \cite{waheed2022comprehensive, duan2020emerging}.
When a vehicle requests task offloading via an RSU, the connected edge server processes the task and returns the computed result.
This architecture is inherently designed to minimize processing delay, which comprises (i) communication delay, defined as the time required to transmit and receive data via RSUs, and (ii) computation delay, which corresponds to the processing time of offloaded tasks from vehicles \cite{mach2017mobile}.

\vspace{-3mm}
\subsection{Motivation, background, and related works}
Vehicular edge computing (VEC) networks aim to identify a task allocation strategy that minimizes the total offloading delay and prevents network congestion. 
This problem, commonly known as vehicle association, has been widely investigated using a number of approaches  \cite{lyu2016multiuser, moubayed2020edge, zhang2019task, fan2025veh, zhang2025drl, fan2024deep}.
Early studies assign vehicular tasks to the nearest edge server \cite{lyu2016multiuser}, whereas subsequent studies incorporate the available computing capacities of edge servers \cite{moubayed2020edge} or offload tasks to remote cloud servers when they provide lower processing latency \cite{zhang2019task}.
The extended association framework introduces cooperative mechanisms among RSUs to address computing result reception failure arising from high vehicle mobility and load imbalance \cite{fan2025veh, zhang2025drl, fan2024deep}. 
In \cite{fan2025veh}, RSU-to-RSU cooperation is integrated into a deep reinforcement learning (DRL)–based offloading scheme to mitigate task-result reception failures when vehicles rapidly move across coverage boundaries. In \cite{zhang2025drl}, under-utilized parked vehicles are leveraged as auxiliary computing nodes to enhance load distribution. A hierarchical orchestration strategy is proposed in \cite{fan2024deep}, where lightweight tasks are executed locally and intensive models are offloaded to RSUs to balance latency and inference accuracy.

Modern VEC systems increasingly rely on multi-core edge servers to support parallel processing \cite{zhang2024survey}. Efficient utilization of multi-core processor architectures, however, introduces a nontrivial task assignment challenge. To fully exploit multi-CPU concurrency, tasks are distributed in a manner that maximizes parallel execution while preventing resource under-utilization and load imbalance \cite{wang2009mapping}.
Several studies have investigated multi-CPU edge computing \cite{zhong2022potam, teng2023game, liang2019multiuser, lee2024distributed, sun2023bargain, feng2021joint}. 
A common strategy focuses on restricting the number of offloaded tasks according to the available CPU count \cite{zhong2022potam, teng2023game}.
In \cite{zhong2022potam}, resource heterogeneity across macro and small base stations is considered to determine the volume of offloaded tasks adaptively. The metadata-first mechanism in \cite{teng2023game} allows the optimal offloading destination to be selected concurrently with data reception.
Furthermore, virtualization techniques are utilized to facilitate parallel execution.
Virtual machines are used to reduce the I/O interference that grows exponentially with the number of concurrently processed tasks \cite{liang2019multiuser} and to further minimize worst-case delay to ensure task reliability \cite{lee2024distributed}. A heterogeneous multi-CPU server architecture is examined in \cite{sun2023bargain}, where load balancing across servers with different CPU capacities is formulated as a stable matching problem. In \cite{feng2021joint}, a CPU-level task partitioning strategy enables concurrent execution of a single task across multiple CPUs. 

To address traffic imbalance in VEC systems, queueing mechanisms have been introduced to prevent task drops when the volume of offloaded tasks exceeds the server capacity \cite{luo2021resource, omar2012evaluation}. 
Most existing offloading strategies model queuing delays under the assumption of unbounded queues \cite{song2021joint, hou2023intelligent, chen2025multi, Wang2024invar}.
In \cite{song2021joint}, inter-RSU connectivity in a heterogeneous VEC architecture is leveraged to minimize the overall delay using evolutionary optimizations. 
In \cite{hou2023intelligent}, the trade-off between energy consumption and queuing delay is explicitly analyzed, while large-scale VEC networks are considered so that vehicles autonomously choose offloading policies and regulate their individual waiting delays \cite{chen2025multi}.
Excessive queuing delays may result in performance inversion, and edge computing becomes slower than remote cloud processing \cite{Wang2024invar}. To mitigate this issue, a heuristic optimization approach is developed in \cite{Wang2024invar}. 
In practice, edge servers maintain queues with finite capacity \cite{BALSAMO2003269} because they can only receive service requests from vehicles located within their geographically constrained coverage regions \cite{kong2022edge}. 
These constraints inherently limit the number of tasks that can be accommodated by each server.

Prior studies on finite-length queues primarily focus on single-CPU systems, where the CPU is assumed to remain continuously occupied \cite{dong2023load, yue2021todg, ma2022drl, peng2024stochastic, pervez2025efficient}. The main objective of these studies is to maintain the queue stability.
Load balancing is used in \cite{dong2023load} to avoid overwhelming individual servers, while \cite{yue2021todg} develops a preemptive strategy that drops outdated tasks to ensure bounded queue length.
Both long-term and short-term queue sizes are managed through DRL that prevents intermittent drop-out of tasks \cite{ma2022drl}. CPU frequency control is explored in \cite{peng2024stochastic} to balance energy consumption and queue length. In \cite{pervez2025efficient}, transmission and computational queues are jointly controlled in resource-limited satellite–aerial networks to prevent computation and communication failures. 
Despite their contributions, these studies rely on a single-CPU processing model and therefore overlook a critical issue of resource under-utilization in multi-CPU environments. When the number of CPUs exceeds the number of assigned tasks, idle processors reduce effective parallelism and contribute to load imbalance, ultimately increasing queuing delay. 

Managing multi-CPU edge servers with finite queues involves determining the optimal number of task assignments to balance CPU under-utilization and congestion. Most existing works simplify the underlying optimization by relaxing discrete task counts into continuous values \cite{zhong2022potam, feng2021joint, pervez2025efficient}.
However, such relaxations fail to capture the intrinsic threshold behavior of queuing systems. 
No queuing delay exists when the number of tasks does not exceed the number of CPUs, but once this tipping point is crossed, delay escalates sharply. 
Since task allocations distort this discontinuous threshold behavior, these approximations are fundamentally limited in identifying the true optimal task assignment.

\begin{table*}[!t]
\centering
\label{tab:comparison}
\renewcommand{\arraystretch}{1.1}
\resizebox{\textwidth}{!}{
\begin{tabular}{c|c|c|c|c|c|c|c}
\hline
Papers &
\shortstack{{Load}\\{Balancing}} &
\shortstack{{Parallel}\\{Processing}} &
\shortstack{{Queueing}\\{Delay}} &
\shortstack{{Resource}\\{Under-utilization}} &
\shortstack{{Finite length}\\{Queue}} &
\shortstack{{Distributed}\\{Algorithm}} &
\shortstack{{Non-relaxed}\\{Algorithm}} \\
\hline
\cite{lyu2016multiuser}  & $-$ & $-$ & $-$ & $-$ & $-$ & $-$ & $-$ \\\hline
\cite{moubayed2020edge} & $-$ & $\checkmark $ & $-$ & $\checkmark $ & $-$ & $-$ & $\checkmark $ \\\hline
\cite{zhang2019task} & $\checkmark $ & $-$ & $-$ & $\checkmark $ & $\checkmark $ & $-$ & $\checkmark $ \\\hline
\cite{fan2025veh,zhang2025drl,fan2024deep} & $\checkmark $ & $-$ & $-$ & $-$ & $-$ & $-$ & $-$ \\\hline
\cite{zhong2022potam} & $\checkmark $ & $\checkmark $ & $-$ & $\checkmark $ & $-$ & $\checkmark $ & $-$ \\\hline
\cite{teng2023game} & $\checkmark $ & $\checkmark $ & $-$ & $\checkmark $ & $-$ & $\checkmark $ & $\checkmark $  \\\hline
\cite{liang2019multiuser}&$\checkmark $ & $\checkmark $ & $-$ & $-$ & $-$ & $-$ & $-$ \\\hline
\cite{lee2024distributed}&$\checkmark $ & $\checkmark $ & $-$ & $-$ & $-$ & $\checkmark $ & $\checkmark $ \\\hline
\cite{sun2023bargain} & $\checkmark $ & $\checkmark $ & $-$ & $\checkmark $ & $-$ & $-$ & $\checkmark $ \\\hline
\cite{feng2021joint} & $\checkmark $ & $\checkmark $ & $-$ & $-$ & $-$ & $\checkmark $ & $-$ \\\hline
\cite{song2021joint,Wang2024invar} & $\checkmark $ & $\checkmark $ & $\checkmark $ & $-$ & $-$ & $-$ & $\checkmark $ \\\hline
\cite{hou2023intelligent,chen2025multi} & $\checkmark $ & $\checkmark $ & $\checkmark $ & $-$ & $-$ & $\checkmark $ & $\checkmark $ \\\hline
\cite{dong2023load,ma2022drl,peng2024stochastic} & $\checkmark $ & $-$ & $\checkmark $ & $-$ & $\checkmark $ & $-$ & $\checkmark $ \\\hline
\cite{yue2021todg} & $\checkmark $ & $-$ & $\checkmark $ & $-$ & $\checkmark $ & $\checkmark $ & $\checkmark $ \\\hline
\cite{pervez2025efficient}& $\checkmark $ & $-$ & $\checkmark $ & $-$ & $\checkmark $ & $\checkmark $ & $-$ \\\hline
{Our Scheme} & $\checkmark $ & $\checkmark $ & $\checkmark $ & $\checkmark $ & $\checkmark $ & $\checkmark $ & $\checkmark $ \\
\hline
\end{tabular}
}
\caption{Comparison with related works ($\checkmark $: yes, $-$: no).}
\end{table*}

\subsection{Contributions}

To overcome this limitation, this work develops a combinatorial offloading strategy that avoids relaxation and effectively balances resource under-utilization and task congestion.
Task assignments are decomposed into subproblems and cooperatively determined through interactions between edge servers and vehicles, implemented via a message-passing (MP) framework inspired by dynamic programming principles \cite{kschischang2001factor}.
Individual queue lengths at edge servers and offloading decisions of vehicles are managed by combinatorially solving these subproblems, with the solutions encapsulated into compact messages.
After the message exchanges, all agents converge to a globally consistent and delay-efficient allocation.


Numerical simulations and a feasibility study conducted in a digital-twin VEC environment based on real-map data demonstrate the superior efficiency of the proposed strategy compared to existing vehicle association algorithms.
Furthermore, rigorous theoretical analysis guarantees both convergence and global optimality of the obtained solution.
These findings highlight the practical potential of the proposed offloading technique for real-world deployment in VEC networks.
The primary contributions of this work are summarized as:
    \begin{itemize}
    \item The waiting delay caused by the imbalance between the number of CPUs and the volume of offloaded tasks is analyzed in vehicular edge servers equipped with multiple CPUs. 
    An optimization problem is formulated to determine the optimal short-term task offloading association that minimizes the overall network queuing delay.

    \item A decentralized allocation mechanism is developed in which RSUs and vehicular users collaboratively determine vehicle associations.
    The convergence and global optimality of the proposed algorithm are theoretically proven and validated through simulations. 
    In addition, its feasibility is assessed using a real-map-based testbed.
    
    \item An in-depth numerical analysis reveals that when task queues are nearly full, predicting the instantaneous remaining computing capacity of edge servers is critical for minimizing waiting delays.
    Instead of approximating the number of tasks that can be processed at once as a continuous value, the proposed technique evaluates waiting delays based on discrete queued tasks to determine the optimal solution.
\end{itemize}

The remainder of the paper is organized as follows: the system model of the parallel edge computing network is presented in Section \ref{sysmodel}. 
Section \ref{allocation} derives the distributed algorithm for the latency minimization target. 
Several technical and theoretical issues for real-world deployment are addressed in Section \ref{theory}. 
Section \ref{results} assesses the impact of scalability, system parameters, and performance of the proposed algorithm. 
Finally, Section \ref{conclusion} concludes the paper.
Table \ref{table:key_notations} summarizes the notations used throughout this work.

\begin{table}
	\centering
	\caption{Notations and descriptions}
	\begin{tabular}{ll}
		\hline
        \hline
		\textbf{Notation} & \textbf{Description}\\
		\hline
        $\mathcal{V}$ & Set of vehicles, indexed by $i$\\
		$\mathcal{A}$ & Set of RSUs, indexed by $a$\\
        $V$ & Number of vehicles\\
        $A$ & Number of RSUs\\
        $c$ & CPU index\\
        $k_a$ & Number of CPUs installed in RSU-$a$\\
        $L_i$ & Size of transmission data of vehicle-$i$\\
        $C_i$ & Required computation resources of vehicle-$i$ \\
        $d_{ia}$ & Distance between vehicle-$i$ and the RSU-$a$\\
        $d_a$ & Circular coverage of RSU-$a$\\
        $B_a$ & System bandwidth of RSU-$a$\\
        $B_p$ & Fixed Bandwidth for vehicle\\
        $h_{ia}$ & Channel power gain from vehicle-$i$ to RSU-$a$\\
        $p_i$ & Transmission power of vehicle-$i$\\
        $\sigma^2$ & Noise power\\
        $r_{ia}$ & Data transmission rate from vehicle-$i$ to RSU-$a$\\
        $f_i^l$ & CPU frequency of vehicle-$i$\\
        $f_a$ & CPU frequency of RSU-$a$\\
        $x_{ia}$ & Offloading decision variables\\
        $n_a$ & Number of tasks associated with RSU-$a$\\
        $N_a$ & Maximum number of tasks associated with RSU-$a$\\
        $n_{ac}$ & Number of tasks associated with CPU-$c$ in RSU-$a$\\
        $T_i^l$ & Local Computing delay of vehicle-$i$\\
        $T_{ia}^o(n_a)$ & Computational delay of vehicle-$i$ offloaded to RSU-$a$\\
        $T_{ia}^c$ & Communication delay from vehicle-$i$ to RSU-$a$\\
        $\mathcal{T}_{ia}(n_a)$ & Overall processing delay from vehicle-$i$ to RSU-$a$\\
        $t^{\mathrm{max}}$ & The longest acceptable delay\\
        $\mathcal{X}_i$ & Sets of variables $x_{ia}, \forall a \in \mathcal{A}$\\
        $\mathcal{X}_a$ & Sets of variables $x_{ia}, \forall i \in \mathcal{V}$\\
        $Q_i(\mathcal{X}_i)$ & Factor function for vehicle-$i$\\
        $R_a(\mathcal{X}_a)$ & Factor function for RSU-$a$\\
        $\mu_{a\rightarrow b}(c=d)$ & Preference that variable $c$ takes on the value of $d$,\\
         &transferred from a node $a$ to $b$\\
        $\alpha_{ia}, \beta_{ia}, \rho_{ia}, \eta_{ia}$& Messages of the factor graph\\
        $\mathrm{rank}^l[X]$ & $l$-th smallest value for a set $X$\\
        $\Psi_{ia}(b)$ & Value for representing each term of $\alpha_{ia}$\\
        $b_{ia}$ & Decision metrics to determine $x_{ia}$\\
		\hline
	\end{tabular}
	\label{table:key_notations}
\end{table}

\section{System Model}
\label{sysmodel}
%
\begin{figure} 
\begin{center}
\includegraphics[width=.9\linewidth]{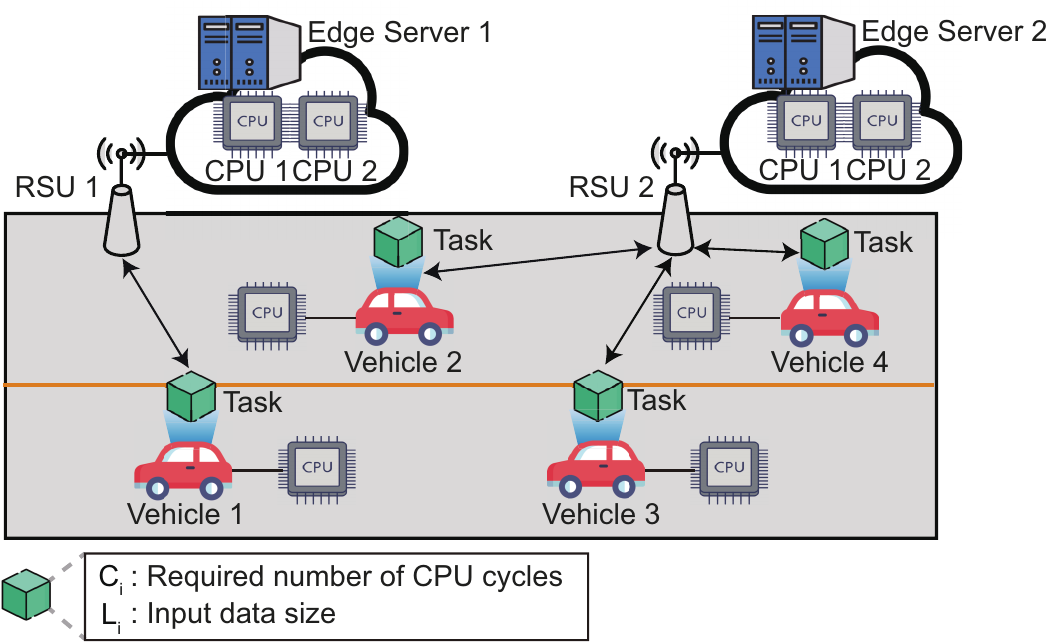}
\caption{Multiprocessor computation vehicular offloading system.} \label{fig1}
\end{center}
\end{figure}



We consider a vehicular network consisting of $A$ RSUs deployed along the roadside and $V$ vehicles, as illustrated in Fig. \ref{fig1}.
Let $\mathcal{A} = \{0, 1, 2, \ldots, A\}$ and $\mathcal{V} = \{1, 2, \ldots, V\}$ denote the index sets of RSUs and vehicles, respectively.
RSU-$a$ refers to the $a$-th RSU for $a \in \mathcal{A} \backslash {0}$, while $a = 0$ represents the on-board unit of a vehicle, indicating local task execution. 
Similarly, vehicle-$i$ refers to the $i$-th vehicle for $i \in \mathcal{V}$. 
Each RSU is equipped with a dedicated edge server capable of executing computational tasks. 
The edge server at RSU-$a$ is equipped with $k_a$ distinct CPUs, each capable of independently processing multiple tasks.
In addition, each vehicle is equipped with a local computing unit for executing vehicular applications such as media or sensing services.
The computational characteristics of a task generated by vehicle-$i$ are defined by its data volume $L_i$ (in bits) and its computational workload $C_i$ (in CPU cycles), with an expected mean value denoted by $\overline{C}_i$.
Each task is processed by a single CPU at a time.
Vehicles generate tasks periodically within a pre-defined time window, denoted by $t^{\mathrm{max}}$, which specifies the maximum permissible time for task completion. 
Since vehicles are inherently both energy and computational resource-limited, they may be unable to complete intensive tasks locally.
In this work, the energy consumption of RSUs is not considered as a design constraint. 
Although multi-CPU edge servers incur higher power consumption from the increased number of active processing units, this cost is offset by the reduced processing delay achieved through parallel processing \cite{shojafar2016energy}. 
Thus, this work mainly focuses on determining whether tasks should be offloaded to RSUs or processed locally to minimize delay and ensure timely service delivery.

In this environment, the vehicle-to-RSU association aims at task allocation in a way that minimizes queuing delays caused by limited CPU resources at the RSUs.
To model this formally, a binary variable $x_{ia}$ is introduced to represent the association status of each vehicle. 
Thus, $x_{ia} = 1$ indicates that vehicle-$i$ is associated with RSU-$a$, whereas $x_{i0} = 1$ denotes local execution on vehicle-$i$. 
The total number of vehicles associated with RSU-$a$ is given by 
$n_a = \sum_{i=1}^V x_{ia}$.
The computation delay for vehicle-$i$ when associated with RSU-$a$, denoted by $\mathcal{T}_{ia}(n_a)$, is given by
\begin{align} \label{obj_delay}
    \mathcal{T}_{ia}(n_a) =&\begin{cases} T_{i}^{\mathrm{loc}} & \text{if~} a = 0,\\
T_{ia}^{\mathrm{off}}(n_a) + T_{ia}^{\mathrm{com}} & \text{otherwise,} \end{cases}
\end{align} 
where $T_{i}^{\mathrm{loc}}$ is the local computing delay of vehicle-$i$, and $T_{ia}^{\mathrm{off}}(n_a)$ is the offloading delay at RSU-$a$, including queuing effects of $n_a$ concurrent tasks. Also, $T_{ia}^{\mathrm{com}} $ is the communication time between vehicle-$i$ and RSU-$a$.
The local processing time for vehicle-$i$, denoted by $T_{i}^{\mathrm{loc}}$, is equal to
\begin{align} \label{local_delay}
T_{i}^{\mathrm{loc}} = \frac{C_{i}}{f_{i}^{\mathrm{loc}}},
\end{align}
where $f_i^{\mathrm{loc}}$ is the CPU frequency of the local computing unit.

When a task is offloaded to RSU-$a$, the total processing delay consists of three components:  transmitting delay, waiting delay, and execution delay. 
If the number of offloaded tasks $n_a$ exceeds the number of available CPUs $k_a$, the excess tasks are queued, leading to additional waiting delay. 
This delay depends on both the number of available CPUs at RSU-$a$, denoted by $k_a$, and the total number of simultaneously assigned tasks, denoted by $n_a$.
The execution time for each task is determined by the ratio of the required CPU cycles $C_i$ for vehicle-$i$ to the CPU frequency $f_a$ of RSU-$a$. 
The resulting execution delay, including queuing effects, is given by
\begin{align} \label{RSU_delay}
    T_{ia}^{\mathrm{off}}(n_a) =  \left(1 - \frac{k_a}{2n_a} \floor{\frac{n_a}{k_a}}\right)\left(1 + \floor{\frac{n_a}{k_a}}\right) \frac{C_i}{f_a}.
\end{align}
\begin{prop}
The number of tasks assigned to CPU-$c$ in RSU-$a$, denoted by $n_{ac}$, is given by
        \begin{align} \label{cpu_num}
            n_{ac} =&\begin{cases} \floor{\frac{n_a}{k_a}}+1 & \text{if~}1\leq c \leq n_a - k_a \floor{\frac{n_a}{k_a}},\\[0.5em]
            \floor{\frac{n_a}{k_a}} & \text{if~} n_a - k_a \floor{\frac{n_a}{k_a}} + 1 \leq c \leq k_a. \end{cases}
        \end{align}
    \end{prop}

Identifying the optimal multi-CPU scheduling configuration is known to be NP-complete \cite{hartmanis1982computers}.
However, previous studies have shown that the performance gain from complex scheduling algorithms is often marginal compared to simple heuristics \cite{stallings2009operating}. 
As a result, load-sharing scheduling, which assigns incoming tasks to CPUs in a round-robin manner, has been widely used in practice \cite{xu2007load, lepers2017towards, hofmeyr2010load, zhuravlev2012survey}.
Under this policy, each CPU initially receives $\floor{\frac{n_a}{k_a}}$ tasks. 
The remaining tasks are then distributed one-by-one to the first few CPUs, in accordance with the pigeonhole principle.
The total number of tasks distributed across CPUs, denoted by $n_{ac}$, satisfies the consistency condition $\sum_{c=1}^{k_a} n_{ac} = n_a$.

\begin{prop}
The total waiting delay for $n_{ac}$ tasks assigned to CPU-$c$ in RSU-$a$, denoted by $T_{\mathrm{all}}^{ac}$, is given by
\begin{align}\label{delay_one_cpu}
    T_{\mathrm{all}}^{ac} &= T_{1}^{ac} + (T_{1}^{ac} + T_{2}^{ac}) + \ldots + (T_{1}^{ac}  + \ldots + T_{n_{ac}-1}^{ac}) \nonumber\\
    &= \sum_{m=1}^{n_{ac}-1} (n_{ac}-m)T_m^{ac}.
\end{align}
\end{prop}
Consider $n_{ac}$ tasks assigned to CPU-$c$ in RSU-$a$, each with a processing time denoted by $T_1^{ac}, \ldots, T_{n_{ac}}^{ac}$. Under sequential execution, the first task begins immediately, experiencing no waiting delay. 
The second task starts after $T_1^{ac}$, the third after $T_1^{ac} + T_2^{ac}$, and so on. 
The total waiting time is the cumulative sum of the processing times of preceding tasks, as in \eqref{delay_one_cpu}.

To compute the expected total delay across all CPUs at RSU-$a$, we make the following assumptions.
The communication time differences between vehicles within the same RSU coverage area are negligible compared to computation times. 
Task arrivals to RSU-$a$ are simultaneous, and $n_a$ tasks can arrive in $n_a!$ possible permutations.
Let ${\tau_m^{ac}}$ be the corresponding sequence of task durations for one such permutation, and let $\tau = \{ (\tau_1^{a1},\ldots,\tau_{n_{a1}}^{a_1}, \tau_{1}^{a_2}, \ldots,  \tau_{n_{k_a}}^{ak_a}) \}$ represent the task permutation across all CPUs in RSU-$a$. 
Then, the expected total delay, denoted by $T_{\mathrm{all}}^a$, is given as
    \begin{align} \label{avgt}
        T_{\mathrm{all}}^a=& \frac{1}{n_a!}\sum_{\tau}\sum_{c=1}^{k_a}\left[\sum_{m=1}^{n_{ac}-1} (n_{ac}-m) \tau_m^{ac} + \sum_{m=1}^{n_{ac}}\tau_m^{ac}\right]\nonumber\\
        =&\frac{1}{n_a!}\left[ \sum_{c=1}^{k_a}\sum_{m=1}^{n_{ac}}(n_{ac}-m+1)\right] \cdot \left[ (n_a-1)!\sum_{c=1}^{k_a}\sum_{m=1}^{n_{ac}}T_m^{ac}\right]\nonumber\\
        =&\left[ \sum_{c=1}^{k_a} \frac{n_{ac}(n_{ac}+1)}{2n_a}\right] \cdot \sum_{c=1}^{k_a}\sum_{m=1}^{n_{ac}}T_m^{ac}\nonumber\\
        =&\left(1 - \frac{k_a}{2n_a} \floor{\frac{n_a}{k_a}}\right)\left(1 + \floor{\frac{n_a}{k_a}}\right) \cdot \sum_{c=1}^{k_a}\sum_{m=1}^{n_{ac}}T_m^{ac}.
\end{align}
Note that $\sum_{c=1}^{k_a}\sum_{m=1}^{n_{ac}}T_m^{ac}$ is the total execution time of tasks offloaded to RSU-$a$, with each $T_m^{ac} = \frac{C_m}{f_a}$. 
Therefore, the expected processing delay for a vehicle is consistent with \eqref{RSU_delay}.

Vehicles transmit their data to RSUs over wireless links. 
Each RSU provides a circular coverage with radius $d_a$, and a vehicle can connect to the RSU upon entering its coverage area. 
The channel gain of the link between vehicle-$i$ and RSU-$a$ is given by $h_{ia} = \zeta (d_{ia})$, which is a function of the distance between vehicle-$i$ and the RSU-$a$, denoted by $d_{ia}$.
Vehicle-$i$ communicates with RSU-$a$ over a fixed bandwidth $B_p$. If RSU-$a$ uses a system bandwidth $B_a$, at most $N_a = \floor{\frac{B_a}{B_p}}$ vehicles can offload their tasks to RSU-$a$.  The data rate of the wireless link with transmit power $p_i$ is given by
\begin{align} \label{rate}
r_{ia} = B_p\log_{2}\left(1 + \frac{p_{i}\left|h_{ia}\right|^2}{\sigma^2}\right),
\end{align}
where $\sigma^2$ is the noise power. 
Owing to the fact that the output size transmitted from the RSU to vehicle $i$ is substantially smaller than the input size uploaded from the vehicle to the RSU, and considering that RSUs typically operate with high transmission power, the resulting transmission delay from RSU $a$ to vehicle $i$ is considered negligible \cite{xu2019energy, chen2015efficient}.
Therefore, the uplink transmission delay from vehicle-$i$ to RSU-$a$ is computed as 
\begin{align} \label{comm_time}
T_{ia}^{\mathrm{com}} = \frac{L_i}{r_{ia}}.
\end{align}
The total processing delay between vehicle-$i$ and RSU-$a$ includes all time delay functions calculated in \eqref{local_delay}, \eqref{RSU_delay}, and \eqref{comm_time} and is expressed as 
\begin{align}\label{obj_delay2}
\mathcal{T}_{ia}(n_a)=\begin{cases} \frac{C_{i}}{f_{i}^{\mathrm{loc}}}   &\text{if~} a = 0,\\
\left(1 - \frac{k_a}{2n_a} \floor{\frac{n_a}{k_a}}\right)\left(1 + \floor{\frac{n_a}{k_a}}\right) \frac{C_i}{f_a} + \frac{L_i}{r_{ia}} &\text{o.w.} \end{cases}
\end{align}

Fig. \ref{fig_sinario} illustrates an example of vehicle association strategy to minimize the delay in \eqref{obj_delay2}. When RSU-$2$ with $2$ CPUs becomes overloaded with $3$ incoming tasks, one of these tasks may be reassigned to RSU-$1$, which has processing capacity.
Such a reassignment increases the processing delay at RSU-$1$ but reduces the overall network delay. 
Thus, task redistribution across RSUs is jointly optimized to minimize the global processing delay.
To formally capture this objective, the task offloading optimization is formulated as 
\begin{subequations} \label{opt1}
\begin{align}
\min_{\{x_{ia}\}} &\quad \sum_{i=1}^V \sum_{a=0}^A \min (\mathcal{T}_{ia}(n_a),t^{\mathrm{max}}_i) x_{ia} \label{obj1} \\
\text{subject to}&\quad \sum_{a=0}^A x_{ia}=1, ~ \forall i \in \mathcal{V}, \label{const1}\\
&\quad \sum_{i=1}^V x_{ia} = n_a, ~ \forall a \in \mathcal{A}\backslash \{0\},  \label{const2} \\
&\quad x_{ia}\in\{0,1\},~ \forall (i,a) \in  \mathcal{V}\times \mathcal{A},\\ 
&\quad n_a \in\{0,\ldots,N_a\},~ \forall a \in \mathcal{A}\backslash \{0\}.
\end{align}
\end{subequations}
\begin{figure} 
\begin{center}
\includegraphics[width=.9\linewidth]{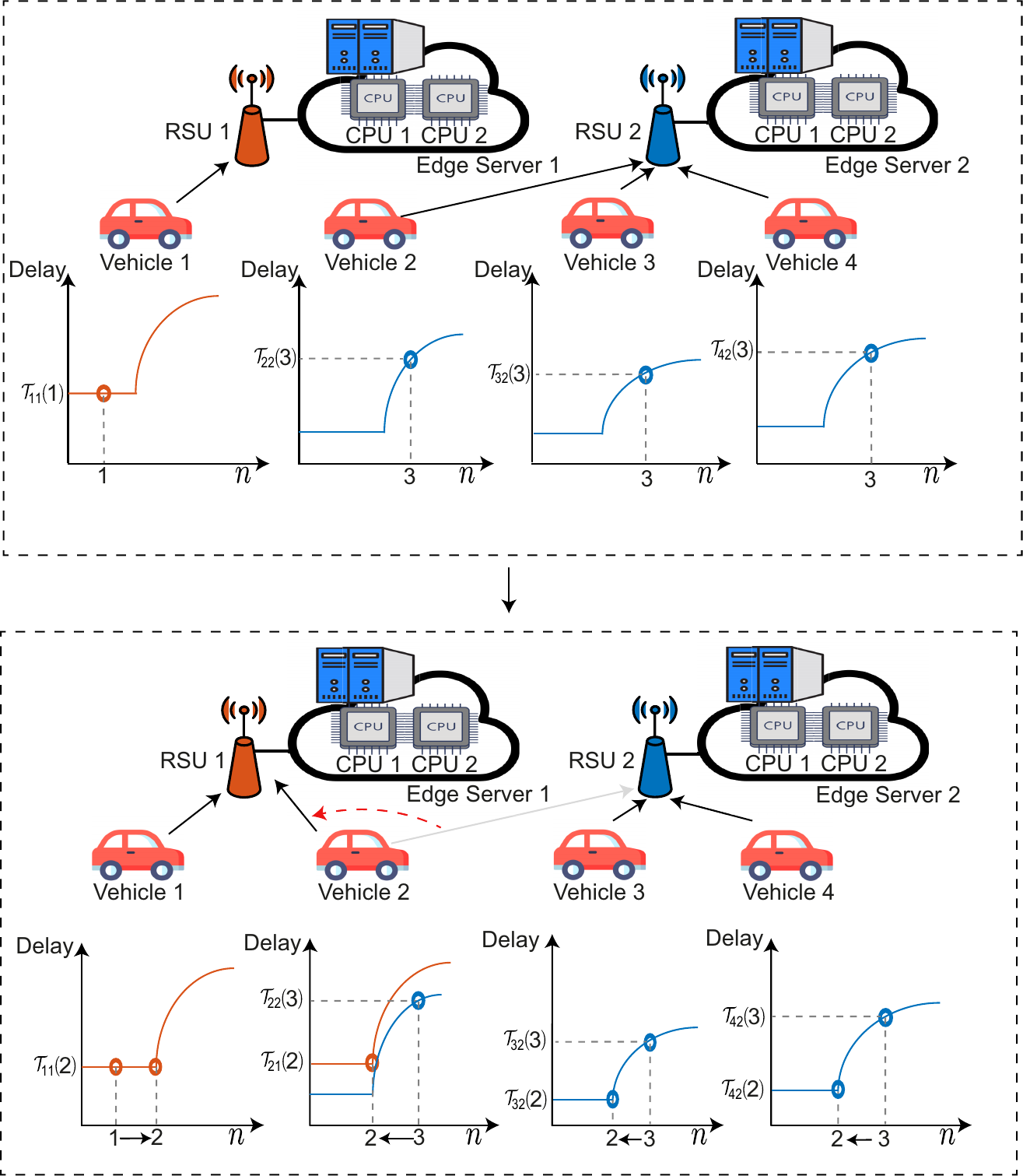}
\caption{Vehicle association for task offloading delay minimization.} \label{fig_sinario}
\end{center}
\end{figure}

Note that the objective function in \eqref{obj1} is bounded by a delay threshold $t^{\mathrm{max}}_i$, which represents the maximum allowable delay within a single scheduling frame. 
If the actual delay exceeds this limit, the task is considered to be in outage.
Furthermore, $t^{\mathrm{max}}_i$ can be adjusted to encode task priorities. Higher-priority tasks are assigned smaller $t^{\mathrm{max}}_i$ values, whereas lower-priority tasks receive larger thresholds. 
This mechanism naturally guides high-priority tasks toward RSUs with improved computational resources.
Constraint \eqref{const1} guarantees that each vehicle selects either one RSU or local computing, while constraint \eqref{const2} defines $n_a$ as the total number of vehicles associated with RSU-$a$.

This optimization problem is inherently dynamic in that the delay function $\mathcal{T}_{ia}(n_a)$ depends on the number of associated tasks. When vehicles autonomously choose their offloading policies, the delay values  evolve according to the global association configuration, rendering centralized global optimization challenging.
Indeed, the problem can be reduced from the $n$-partition problem, where $V$ vehicles are partitioned into $A+1$ disjoint groups of no more than $N_a$ vehicles each, while attempting to match the resulting delay in each group to its optimal delay configuration. 
This problem is known to be NP-hard \cite{garey2002computers}, indicating that obtaining the optimal solution is computationally demanding.

To this end, an alternative distributed optimization strategy is developed.
Instead of jointly optimizing the full network configuration in a centralized manner, \eqref{opt1} is decomposed into subproblems corresponding to individual RSUs and vehicles and resolved by leveraging local information of individual agents.
Each RSU is responsible for managing its own task queue by determining which and how many vehicles should be associated with it, while vehicles select their offloading targets based on local utility metrics such as delay or queue length. 
The configuration of the entire network is then achieved through negotiation or coordination among these local offloading policies to ensure the global consistency without requiring complete knowledge. 
Furthermore, this approach reduces the computation loads associated with centralized optimization and scales more effectively in large-scale VEC environments. 
Thus, a distributed algorithm becomes a natural and scalable choice by enabling network-wide parallelism without relying on centralized cloud with intensive computational capabilities.

\section{Distributed Algorithm}
\label{allocation}
This section develops a distributed algorithm via an MP framework to find the optimal vehicle association. 

\vspace{-3mm}
\subsection{Formulation} \label{vehicle section}
The design of a distributed solution aims to employ the principle of dynamic programming, whereby a global optimization problem is solved through the combination of optimal solutions to its constituent subproblems. 
Thus, the original problem in \eqref{opt1} is decomposed into subproblems assigned to two agent groups of vehicles and RSUs so that each locally addresses a designated subproblem. 
When individual subproblem solutions are feasible and consistent, their combination yields a globally optimal solution to \eqref{opt1} \cite{givoni2009binary}.

To facilitate this decomposition, each subproblem is associated with a factor function \cite{kschischang2001factor}, which is handled either by a vehicle or an RSU.
To implement this principle, each subproblem is necessarily associated with a local function that can be handled either by a vehicle or an RSU. 
This local function is referred to as a factor function \cite{kschischang2001factor} since the collection of factor functions can reconstruct the original problem in \eqref{opt1}. The set of all factor functions together reconstructs the original optimization objective.
For convenience of notation, define $\mathcal{X}_{a} \equiv \{x_{ia}: i \in \mathcal{V}\}$ as the set of association variables related to RSU-$a$, and $\mathcal{X}_{i} \equiv \{x_{ia}: a \in \mathcal{A}\}$ for vehicle-$i$. Vehicle-$i$ is responsible for selecting exactly one RSU or local execution. The corresponding constraint in \eqref{const1} is represented in a factor function defined as 
\begin{align} 
Q_i(\mathcal{X}_{i})=&\begin{cases}\infty & \text{if~}\sum_{a \in \mathcal{A}}x_{ia} \neq 1,\\
0 &\text{otherwise.}\end{cases}\label{factor1}
\end{align}
Meanwhile, RSU-$a$ governs its own queue and seeks to minimize the total processing delay under a constraint on the number of associated vehicles.
The factor function models the objective function in \eqref{obj1} and the constraint in \eqref{const2} as
\begin{align} 
R_a(\mathcal{X}_{a})=&\begin{cases}\infty ~~~~~~~ &\text{if~}\sum_{i \in \mathcal{V}}x_{ia} \neq n_{a}, \\
\sum\limits_{i \in \mathcal{V}} x_{ia} \min(\mathcal{T}_{ia}(n_a),t^{\mathrm{max}}) &\text{otherwise.}\end{cases}\label{factor2}
\end{align}
The global optimization problem in \eqref{opt1} is now equivalently expressed in the following unconstrained form as
\begin{align}\label{opt2}
\underset{\{x_{ia}\}}{\text{min}} \sum_{i \in \mathcal{V}}Q_i(\mathcal{X}_{i})+\sum_{a \in \mathcal{A}}R_a(\mathcal{X}_{a}).
\end{align}
Note that each vehicle or RSU aims to minimize its own factor function. A globally optimal solution is achieved only if all factor functions are simultaneously minimized and mutually consistent. If any factor function evaluates to $\infty$, the corresponding local configuration is infeasible, and thus, the global objective cannot be minimized.

However, typical variable assignments that minimize individual factor functions may conflict across different agents.
To resolve such inconsistencies, a cooperative mechanism that drives the global agreement is developed through the MP framework \cite{kschischang2001factor}.  
Each vehicle or RSU encodes its local solution in a message, represented as a real-valued function, and exchanges it with the neighborhood of the factor graph. 
If local solutions differ between two-factor functions, their local solutions are adjusted in response to received messages until global agreement (consensus) is achieved. Once convergence is reached, the final configuration satisfies all constraints and minimizes the global objective. A rigorous convergence and optimality analysis is provided in the next section.

To visualize the distributed computation, a factor graph \cite{kschischang2001factor} is introduced, which depicts the message flows between vehicles and RSUs.
In this bipartite graph, circular nodes represent variables corresponding to the association status between a vehicle and an RSU, and square nodes represent factor functions associated with individual constraints and objectives, while edges connect variables to associated factor functions. Fig.~\ref{fig3} illustrates a factor graph for a network example with four vehicles and two RSUs in Fig.~\ref{fig1}.
The  next subsection provides a detailed derivation of message update rules to compute the globally optimal association.

\begin{figure} 
\begin{center}
\subfigure[Factor graph]{
\includegraphics[width=.75\linewidth]{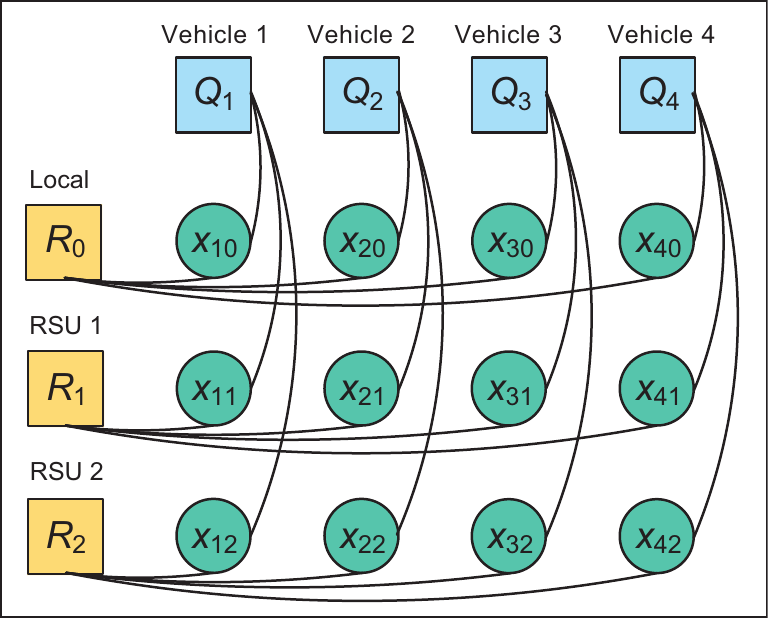}\label{fig31}
}
\subfigure[Message definitions]{
\includegraphics[height=.12\linewidth]{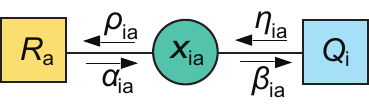}\label{fig32}
}
\caption{The factor graph for an example with 4 vehicles and 2 RSUs.} \label{fig3}
\end{center}
\end{figure} 

\subsection{Derivation}
The derivation of message update rules is presented here. A message is generated for every edge in the factor graph (see Fig.~\ref{fig31}). 
Note that each edge connects a variable node $x_{ia}$ with two adjacent factor nodes. Two messages are defined for both directions of the edge.
For each edge in the associated factor graph, a message is defined to represent the preference for the corresponding binary variable value. To be precise, the message $\mu_{a \rightarrow b}(c = d)$ indicates the preference that variable $c$ takes the value $d$, sent from node $a$ to node $b$.

In the context of the vehicle association problem minimizing the overall processing delay as in \eqref{opt1}, these messages encapsulate the cost of associating vehicle-$i$ with RSU-$a$ in terms of delay. 
Since variable $x_{ia}$, representing vehicle-$i$ assigned to RSU-$a$, is involved in two factor functions, namely $Q_i(\cdot)$ for the vehicle-side constraint and $R_a(\cdot)$ for the RSU-side objective, four distinct messages are defined around $x_{ia}$. 
To simplify computation, instead of maintaining messages for both values of binary variable $x_{ia} \in \{0,1\}$, a scalar representation is used with their difference. For each pair $(i,a)$, four messages are defined as
\begin{align}\label{mess_def1}
\beta_{ia}=&\mu_{x_{ia} \rightarrow Q_i}(x_{ia}=1)-\mu_{x_{ia} \rightarrow Q_i}(x_{ia}=0),\nonumber \\
\eta_{ia}=&\mu_{Q_i \rightarrow x_{ia}}(x_{ia}=1)-\mu_{Q_i \rightarrow x_{ia}}(x_{ia}=0),\nonumber \\
\rho_{ia}=&\mu_{x_{ia} \rightarrow R_a}(x_{ia}=1)-\mu_{x_{ia} \rightarrow R_a}(x_{ia}=0),\nonumber \\
\alpha_{ia}=&\mu_{R_a \rightarrow x_{ia}}(x_{ia}=1)-\mu_{R_a \rightarrow x_{ia}}(x_{ia}=0).
\end{align}
The directions of these messages are illustrated in Fig.~\ref{fig32}. Since these messages are difference-based, their signs carry semantic meaning: a positive value indicates that $x_{ia} = 1$, i.e., association, incurs a higher cost than $x_{ia} = 0$, thereby discouraging the association.

The derivation of update rules begins with messages from variable nodes, i.e., $\beta_{ia}$ and $\rho_{ia}$. Each outgoing message equals the sum of all incoming messages except from the destination \cite{yedidia2011message}. 
Since each variable node connects to two factor nodes, its outgoing messages at time instant $t$ are given by
\begin{align}
\rho_{ia}^{(t+1)}=&\eta_{ia}^{(t)},\nonumber\\
\beta_{ia}^{(t)}=&\alpha_{ia}^{(t)}.\label{varmessage1}
\end{align}
Next, messages from function nodes are derived using the min-sum computation rule \cite{kschischang2001factor}. This  implements dynamic programming to minimize additive objectives associated by factor functions. In particular, the factor function $Q_i(\mathcal{X}_i)$, which encodes the constraint that vehicle-$i$ must select exactly one RSU, is considered first. Since the corresponding message $\eta_{ia}$ is defined as the difference between $\mu_{Q_i \rightarrow x_{ia}}(x_{ia}=1)$ and $\mu_{Q_i \rightarrow x_{ia}}(x_{ia}=0)$, these two values are computed by minimizing the sum of the local factor function and the incoming messages from adjacent variable nodes. When $x_{ia}=1$, the corresponding message is defined as
\begin{align}
&\mu_{Q_i \rightarrow x_{ia}}(x_{ia}=1)\nonumber\\
&=\min_{\mathcal{X}_{a}\backslash x_{ia}}\Big(Q_i(x_{ia}=1,\mathcal{X}_{a}\backslash x_{ia})+\sum_{j \in \mathcal{A} \backslash a}\mu_{x_{ij} \rightarrow Q_i}(x_{ik})\Big),
\end{align}
with respect to variables in $\mathcal{X}_i$ except $x_{ia}$. A similar expression holds for $x_{ia}=0$. This procedure resembles the add-compare-select operation in the Viterbi algorithm \cite{kschischang2001factor}. The outgoing message $\eta_{ia}$ is then the difference expressed as
\begin{align}
\eta_{ia}^{(t)}&=-\min_{b \in \mathcal{A} \backslash a} \beta_{ib}^{(t)}, \label{etamessage}
\end{align}
which is obtained as 
\begin{subequations}\label{etamessage1}
\begin{align}
\eta_{ia}^{(t)}=&\mu_{Q_i \rightarrow x_{ia}}(x_{ia}=1)-\mu_{Q_i \rightarrow x_{ia}}(x_{ia}=0) \\
=&\min_{\mathcal{X}_{a}\backslash x_{ia}}\Big(Q_i(x_{ia}=1,\mathcal{X}_{a}\backslash x_{ia})+\sum_{j \in \mathcal{A} \backslash a}\mu_{x_{ij} \rightarrow Q_i}(x_{ik})\Big)\nonumber\\
-& \min_{\mathcal{X}_{a}\backslash x_{ia}}\Big(Q_i(x_{ia}=0,\mathcal{X}_{a}\backslash x_{ia})+\sum_{j \in \mathcal{A} \backslash a}\mu_{x_{ij} \rightarrow Q_i}(x_{ik})\Big) \label{eta1}\\
=&\sum_{j \in \mathcal{A} \backslash a}\mu_{x_{ij} \rightarrow Q_i}(0)\nonumber\\
-&\min_{b \in \mathcal{A} \backslash a}\Big(\mu_{x_{ib} \rightarrow Q_i}(1)+\sum_{j \in \mathcal{A} \backslash\{a,b\} }\mu_{x_{ij} \rightarrow Q_i}(0)\Big)\label{eta2}\\
=&-\min_{b \in \mathcal{A} \backslash a}\Big(\mu_{x_{ib} \rightarrow Q_i}(1)-\mu_{x_{ib} \rightarrow Q_i}(0)\Big)\label{eta3}\\
=&-\min_{b \in \mathcal{A} \backslash a} \beta_{ib}^{(t)}.\label{eta4}
\end{align}
\end{subequations}
Note that in \eqref{eta1}, two minimizations are calculated over feasible assignments of the other variables. 
The factor function $Q_i(\cdot)$ ensures that only configurations with a single active variable are permitted.
For the first term, since $x_{ia}=1$, all others must be 0, yielding a unique configuration, while the second term involves configurations where one of the other variables, i.e., $x_{ib}=1$ is active, and all others are zero. As shown in \eqref{eta2}, both minimizations share a common summation over zero-valued incoming messages, which cancel out, leading to \eqref{eta3}. The final simplification in \eqref{eta4} follows directly from the definition of $\beta_{ib}^{(t)}$. This result implies that each vehicle computes the outgoing message to RSU-$a$ based on the best alternative RSU in terms of minimal cost, which enforces the exclusivity constraint that each vehicle may associate with only one RSU.

Combining this result with the update rules in \eqref{varmessage1} leads to a concise update rule for $\rho_{ia}^{(t+1)}$ as
\begin{align} \label{rhomessage}
\rho_{ia}^{(t+1)} =-\min_{j \in \mathcal{A} \backslash a} \alpha_{ij}^{(t)}.
\end{align}
The physical interpretation of \eqref{rhomessage} is briefly provided.  

Each $\alpha_{ij}^{(t)}$ reflects the reluctance or willingness of RSU-$j$ to accept offloading from vehicle-$i$, as determined by its current load and processing latency. If all incoming messages are positive, $\rho_{ia}^{(t+1)}$ becomes negative, which indicates that no other RSUs prefers to associate with vehicle-$i$ since the processing cost is positive. 
This forces RSU-$a$ to connect to vehicle-$i$ as a time cost saving, i.e., a negative message value. 
For any negative message, $\rho_{ia}^{(t+1)}$ becomes positive, i.e., another RSU is offering lower-cost processing conditions. Thus, the incentive for
vehicle-$i$ to associate with RSU-$a$ is reduced or even discouraged.

Subsequently, the message computation rule for $\alpha_{ia}^{(t)}$ is derived. This message quantifies the marginal latency cost on the objective function in \eqref{opt1} when vehicle-$i$ is linked to RSU-$a$. Since $R_a(\mathcal{X}_a)$ depends explicitly on the number of associated vehicles $n_a$, deriving this message involves enumerating all possible values of $n_a$. For each possible $n_a \in \{0,\ldots, N_a\}$, the expected 
cost of vehicle-$i$ joining RSU-$a$ is evaluated using the current set of incoming messages $\{\rho_{ia}^{(t)}\}$, and the minimum cost is selected. This procedure yields the message update as
\begin{align} \label{alphamessage}
\alpha_{ia}^{(t)}=&~ \Psi_{ia}(1)-\left(\Psi_{ia}(0)\right)^{-},
\end{align}
where $(\cdot)^{-} \equiv \min(0,\cdot)$, and the function $\Psi_{ia}(b)$ is defined with $\{\rho_{ja}^{(t)}:j \in \mathcal{V} \backslash i\}$ and an input variable $b$ as
\begin{align}
&\Psi_{ia}(b) \overset{\Delta}{=}\min_{1\leq n \leq N_a}\Big(b\mathcal{T}_{ia}(n)+\sum_{l=1}^{n-b}
\underset{j \in \mathcal{V} \backslash i}{\text{rank}^{l}}[\mathcal{T}_{ja}(n)+\rho_{ja}^{(t)}]\Big).\label{Amessage1}
\end{align}
Here, $\text{rank}^{l}[X]$ stands for the $l$-th smallest value in a set $X$. 
The function represents the estimated total cost incurred at RSU-$a$ when vehicle-$i$ either joins the server ($b=1$) or not ($b=0$). Thus, the term $b\mathcal{T}_{ia}(n)$ corresponds to the latency cost of offloading for vehicle-$i$, and the following ranked summation predicts the cost contributions of the remaining $n-b$ vehicles expected to be selected by RSU-$a$.
Since the former increases and the latter decreases as $n$ increases, their sum is expected to minimize the cost at an optimal $n$. Such a value of  $n$ corresponds to the optimal queue length balancing congestion and under-utilization.
The message $\alpha_{ia}^{(t)}$ thus represents the differential cost of assigning vehicle-$i$ to RSU-$a$.
A positive value implies that vehicle-$i$ increases congestion and processing delay at RSU-$a$, while a negative value indicates a net saving or load-balancing benefit. The full derivation of $\alpha_{ia}^{(t)}$ is presented in \eqref{alphamessage2}.
\begin{figure*}[t]
\normalsize
\begin{align} \label{alphamessage2}
\alpha_{ia}^{(t)}=&\mu_{R_a \rightarrow x_{ia}}(x_{ia}=1)-\mu_{R_a \rightarrow x_{ia}}(x_{ia}=0) \nonumber\\
=&\min_{\mathcal{X}_{i}\backslash x_{ia}} \Big(R_a(x_{ia}=1,\mathcal{X}_{i}\backslash x_{ia})+\sum_{j \in \mathcal{V} \backslash i}\mu_{x_{ja} \rightarrow R_a}(x_{ja})\Big)-\min_{\mathcal{X}_{i}\backslash x_{ia}} \Big(R_a(x_{ia}=0,\mathcal{X}_{i}\backslash x_{ia})+\sum_{j \in \mathcal{V} \backslash i}\mu_{x_{ja} \rightarrow R_a}(x_{ja})\Big) \nonumber \\
=&\min_{\mathfrak{X}_{i}\backslash x_{ia}}\Big(\mathcal{T}_{ia}(1)+\sum_{j \in \mathcal{V} \backslash i} \mu_{x_{ja} \rightarrow R_a}(x_{ja}=0), \nonumber \\
&\quad\mathcal{T}_{ia}(2)+\min_{j \in \mathcal{V} \backslash i} \Big(\mathcal{T}_{ja}(2)+\mu_{x_{ja} \rightarrow R_a}(x_{ja}=1)+\sum_{c \in \mathcal{V} \backslash \{i,j\}} \mu_{x_{ca} \rightarrow R_a}(x_{ca}=0)\Big),\ldots, \nonumber \\
&\quad\mathcal{T}_{ia}(N_a)+\sum_{k=1}^{N_a-1} \underset{j \in \mathcal{V} \backslash i}{\text{rank}^{k}}\Big[\mathcal{T}_{ja}(N_a)+\mu_{x_{ja} \rightarrow R_a}(x_{ja}=1)+\sum_{c \in \mathcal{V} \backslash \{i,j\}} \mu_{x_{ca} \rightarrow R_a}(x_{ca}=0)\Big]\Big)\nonumber \\
&-\min_{\mathfrak{X}_{i}\backslash x_{ia}}\Big(\sum_{j \in \mathcal{V} \backslash i} \mu_{x_{ja} \rightarrow R_a}(x_{ib}=0),\min_{j \in \mathcal{V} \backslash i} \Big(\mathcal{T}_{ja}(1)+\mu_{x_{ja} \rightarrow R_a}(x_{ja}=1)+\sum_{c \in \mathcal{V} \backslash \{i,j \}} \mu_{x_{ca} \rightarrow R_a}(x_{ca}=0)\Big), \nonumber \\
&\quad\sum_{k=1}^2 \underset{j \in \mathcal{V} \backslash i}{\text{rank}^{k}}\Big[\mathcal{T}_{ja}(2)+\mu_{x_{ja} \rightarrow R_a}(x_{ja}=1)+\sum_{c \in \mathcal{V} \backslash \{i,j\}} \mu_{x_{ca} \rightarrow R_a}(x_{ca}=0)\Big],\ldots,\nonumber \\
&\quad\sum_{k=1}^{N_a} \underset{j \in \mathcal{V} \backslash i}{\text{rank}^{k}}\Big[\mathcal{T}_{ja}(N_a)+\mu_{x_{ja} \rightarrow R_a}(x_{ja}=1)+\sum_{c \in \mathcal{V} \backslash \{i,j\}} \mu_{x_{ca} \rightarrow R_a}(x_{ca}=0)\Big]\Big) \nonumber \\
=&\min_{\mathfrak{X}_{ia}}\Big(\mathcal{T}_{ia}(1),
\mathcal{T}_{ia}(2)+\min_{j \in \mathcal{V} \backslash i} (\mathcal{T}_{ja}(2)+\rho_{ja}),
\ldots,\mathcal{T}_{ia}(N_a)+\sum_{l=1}^{N_a-1} \underset{j \in \mathcal{V} \backslash i}{\text{rank}^{l}}[\mathcal{T}_{ja}(N_a)+\rho_{ja}]\Big)\nonumber \\
&-\min_{\mathfrak{X}_{i}\backslash x_{ia}}\Big(0,
\min_{j \in \mathcal{V} \backslash i} (\mathcal{T}_{ja}(1)+\rho_{ja}),
\ldots,\sum_{l=1}^{N_a} \underset{j \in \mathcal{V} \backslash i}{\text{rank}^{l}}[\mathcal{T}_{ja}(N_a)+\rho_{ja}]\Big) \nonumber \\
=&\Psi_{ia}(1)-(\Psi_{ia}(0))^{-}. 
\end{align}
\hrulefill
\end{figure*}

Using this result and the update rule from \eqref{rhomessage}, each vehicle autonomously determines the utility of associating with a particular RSU through the metric defined as
\begin{align}
p_{ia}^{(t)}=\alpha_{ia}^{(t)} + \rho_{ia}^{(t)}. \label{bvalue}
\end{align}
The RSU that minimizes $p_{ia}^{(t)}$ is selected by vehicle-$i$ for task offloading. This assignment reflects both the willingness of RSU-$a$ serving vehicle-$i$ (via $\alpha_{ia}^{(t)}$) and the relative preference of vehicle-$i$ for RSU-$a$ (via $\rho_{ia}^{(t)}$). The resulting distributed procedure ensures that each RSU and vehicle iteratively exchange scalar messages, leveraging local information and message updates to collectively minimize the global latency. Upon convergence of all messages, the resulting configuration satisfies both vehicle-to-RSU exclusivity and RSU capacity constraints while minimizing overall processing latency. The complete procedure is summarized in Algorithm \ref{algo0}.


\begin{algorithm}
\caption{Distributed vehicular association algorithm} \label{algo0}
\begin{algorithmic}[1]
\STATE {\textbf{Input: }} $C_i, L_i, f_i^l, k_a, f_a, T_{ia}^c$.
\STATE {\textbf{Output: }} Vehicle-RSU association variable set $\{x_{ia}\}$.
\STATE {} Initialize $\rho_{ia}^{(1)}=0$ for all $(i,a)$ and $t = 1$.
\STATE {\textbf{Repeat}}
\STATE {}  ~Use \eqref{rhomessage} to update message $\rho_{ia}^{(t+1)}$ at vehicle-$i$ for $a\in \mathcal{A}$ and send it back to RSU-$a$.
\STATE {}  ~Use \eqref{alphamessage} to update message $\alpha_{ia}^{(t)}$ at RSU-$a$ for $i\in \mathcal{V}$ and send it back to vehicle-$i$.
\STATE {} ~Set $t\leftarrow t+1$.
\STATE {\textbf{Until}} All messages converge or $t$ reaches the limit
\STATE {} Set $x_{ia^{\star}} = 1$ to the pair $(i,a^{\star})$ with the smallest $m_{ia^*}^{(t)}$; set $x_{ia} = 0$ for other $a$.
\end{algorithmic}
\end{algorithm}

\subsection{Technical issues}
The computational complexity of the proposed distributed algorithm is first analyzed. Each vehicle processes at most $O(A)$ incoming messages. In particular, the computation of $\alpha_{ia}^{(t)}$ is the most demanding, as it requires identifying $N_a$ minimum values among incoming messages. This operation involves a sorting step with complexity $O(N_a \log{N_a})$. Since each vehicle may need to send such messages to all $A$ candidate RSUs, the per-vehicle complexity becomes $O(A N_a \log{N_a})$. 
In the worst case, each RSUs may accommodate all vehicles, i.e., $N_a=O(V)$, resulting in a per-iteration complexity upper-bounded by $O(AV \log{V}) = O(V^2\log{V})$. 
Simulation results further show that the algorithm typically converges within $10$ iterations, ensuring manageable overall complexity in practical deployments. 

Beyond computational load, energy efficiency is also evaluated. According to the well-established energy models in \cite{heinzelman2002application}, the computation energy per bit is $50 ~\text{nJ/bit}$, while transmission energy is approximately $10 ~\text{pJ/bit/}\text{m}^2$.Empirical validation shows that five-bit resolution for message representation suffices for accurate MP computations. As a result, the total energy consumption per message exchange is approximately $3.37 ~\mu\text{J}$. Compared to the communication energy consumed during regular RSU-to-vehicle interactions, this computation and messaging overhead is minimal. 

Mobility-induced service interruption is another issue in VEC environments. 
When a vehicle moves out of the RSU communication range before receiving the processed task result, the associated information may be lost, degrading service reliability. 
The proposed framework addresses this issue by adjusting the allowable maximum delay $t_i^{\mathrm{max}}$. 
Vehicles predicted to exit the coverage area soon are assigned small $t_i^{\mathrm{max}}$ values, encouraging the algorithm to offload their tasks to RSUs with higher computational capability so that processing completes before the vehicle leaves the range. 
If timely completion is impossible regardless of assignment, the framework automatically selects local computation to ensure task success.
This mechanism enhances robustness against mobility-induced failures even under high mobility.
These results efficiently demonstrate that the proposed algorithm imposes only modest loads, thereby supporting its practical feasibility for large-scale VEC networks.

\section{Theoretical analysis}
\label{theory}
This section provides a theoretical analysis of the convergence and optimality of the proposed distributed algorithm. 
The convergence is established through the theory of non-expansive mapping \cite{bertsekas1997nonlinear}. 
\begin{definition}
Let $\boldsymbol{y}$ and $\boldsymbol{z}$ be any two input vectors. A mapping $\mathbb{T}$ is said to be non-expansive if there exists a constant $\delta \in (0,1)$ such that
\begin{align}
    &\quad \|\mathbb{T}(\boldsymbol{y}) - \mathbb{T}(\boldsymbol{z})\|_{\infty} \leq \delta \|\boldsymbol{y} - \boldsymbol{z}\|_{\infty}, \quad \forall \boldsymbol{y}, \boldsymbol{z},
\end{align}    
where $\|\mathbb{T}(\boldsymbol{y}))\|_{\infty} = \max_{(i,a)}|\mathbb{T}_{ia}(\boldsymbol{y})|$, i.e., a non-expansive mapping is a function that yields the corresponding sequence that converges to a fixed limit for arbitrary initialization. 
\end{definition}
If the relationship between messages obtained at adjacent time instants maintains a non-expansive property, the algorithm is guaranteed to have a unique fixed point \cite{bertsekas1997nonlinear}. 
\begin{theorem} \label{theorme1}
An iterative algorithm defined by a non-expansive mapping converges to a unique fixed point. 
\end{theorem}
Let $\boldsymbol{y}^{t}$ and $\boldsymbol{z}^{t}$ be two distinct collections of the messages ${\alpha_{ia}^{(t)}}$ at the $t$-th iteration. Define the update mapping $\mathbb{T}$ by $\boldsymbol{y}^{t+1} = \mathbb{T}(\boldsymbol{y}^{t})$ by plugging \eqref{rhomessage} into \eqref{alphamessage}. Let $\mathbb{F}(\boldsymbol{y}^{t})$ and $\mathbb{F}(\boldsymbol{z}^{t})$ denote the collection of the corresponding message mapping functions from the RSU-$a$ to vehicle-$i$ as $\mathbb{F}(\boldsymbol{y}^{t}) = [\mathbb{F}_{ia}(\boldsymbol{y}^{t})]$ and $\mathbb{F}(\boldsymbol{z}^{t}) = [\mathbb{F}_{ia}(\boldsymbol{z}^{t})]$, respectively.  
Thus, a single iteration is expressed as $y_{ia}^{t+1} = \mathbb{T}_{ia}(\boldsymbol{y}^t) = -\min_{b \in \mathcal{A} \backslash a} \mathbb{F}_{ib}(\boldsymbol{y}^t)$. 

For two messages $\alpha_{ia}(x_{ia} = 1)$ and $\alpha_{ia}(x_{ia} = 0)$, $n_a$ and $\Bar{n}_a$ denote the counts of the affiliated vehicles at the RSU-$a$, respectively, for $\boldsymbol{y}^{t}$. Likewise, $m_a$ and $\bar{m}_a$ are defined for $\boldsymbol{z}^{t}$ incomes. The case where $n_a > \bar{n}_a$ and $m_a > \bar{m}_a$ is considered, while the remaining three cases are addressed similarly. Let $v_{k_1}$ and $w_{l_1}$ be $k_1$-th and $l_1$-th smallest indices of $\mathcal{T}_{v_k a}(n_a) + y_{v_k a}^t$ and $\mathcal{T}_{w_l a}(m_a) + z_{w_l a}^t$. Then, $\mathbb{F}_{ia} (\mathbf{y}^t)$ is given by
\begin{align} \label{conv_def_3}
\mathbb{F}_{ia} (\mathbf{y}^t)&=\mathcal{T}_{ia}(n_a)+ \sum_{k=1}^{n_a}(\mathcal{T}_{v_k a}(n_a) + y_{v_k a}^t)\nonumber\\
&- \sum_{k=1}^{\bar{n}_a} (\mathcal{T}_{\bar{v}_k a}(\bar{n}_a) + y_{\bar{v}_k a}^t).
\end{align}
By Theorem \ref{theorme1}, if the output message difference is smaller than the input message difference, the corresponding mapping function is non-expansive. When $n_a - \bar{n}_a > m_a - \bar{m}_a$, an upper bound for the difference $|\mathbb{F}_{ia} (\mathbf{y}^t)-\mathbb{F}_{ia} (\mathbf{z}^t)|$ is obtained by the following inequality.
\begin{theorem} The output message difference $\mathbb{F}_{ia} (\mathbf{y}^t)-\mathbb{F}_{ia} (\mathbf{z}^t)$ is bounded by 
    \begin{align} \label{conv_theorem_2}
        (m_a - \bar{m}_a)&(y_{v_{n_a},a}^t - z_{w_{\bar{m}_a+1,a}}^t) + \epsilon_0\leq \mathbb{F}_{ia}(\boldsymbol{y^{t}})-\mathbb{F}_{ia}(\boldsymbol{z^{t}}) \nonumber\\
        &\leq (m_{a}-\bar{m}_{a})(y_{v_{\bar{n}_a+1, a}}^t- z_{w_{\bar{m}_a+1,a}}^t) + \epsilon_0,
    \end{align}
where $\epsilon_0$ represents negligible higher-order terms.
\end{theorem}
\begin{proof}
\begin{figure*}[!t]
\begin{subequations} \label{mapping_T2}
\begin{align}
\mathbb{F}_{ia}(\boldsymbol{y^{t}})-\mathbb{F}_{ia}(\boldsymbol{z^{t}}) & \leq (m_{a}-\bar{m}_{a})(y_{v_{\bar{n}_a+1, a}}^{t}- z_{w_{\bar{m}_a+1, a}}^{t}) \label{mapping_T2_1} \\
&\quad +\Big( (n_{a}-\bar{n}_{a})-(m_{a}-\bar{m}_{a})\Big)\Big(\mathcal{T}_{v_{\bar{n}_a+1,a}}(n_a) + y_{v_{\bar{n}_a+1,a}}^{t}\Big) \label{mapping_T2_2} \\
&\quad+\mathcal{T}_{ia}(n_a)-\mathcal{T}_{ia}(m_a)+(m_{a}-\bar{m}_{a})(\mathcal{T}_{v_{\bar{n}_a+1,a}}(n_a)-\mathcal{T}_{w_{\bar{m}_a+1,a}}(m_a)) \label{mapping_T2_3} \\
&\quad+\bar{n}_{a}\Big(\mathcal{T}_{v_{1}a}(n_a)-\mathcal{T}_{\bar{v}_{1}a}(\bar{n}_a) + y_{v_1 a}^t - y_{\bar{v}_1 a}^t\Big)-\bar{m}_{a}\Big(\mathcal{T}_{w_1 a}(m_a)-\mathcal{T}_{w_1 a}(\bar{m}_{a}) + z_{w_1 a}^t - z_{\bar{w}_1 a}^t\Big)
\label{mapping_T2_4} \\
&=(m_{a}-\bar{m}_{a})(y_{v_{\bar{n}_a+1, a}}^t- z_{w_{\bar{m}_a+1,a}}^t)+ \epsilon_0
.\label{mapping_T2_5}
\end{align}
\end{subequations}
\hrulefill
\end{figure*}    
The difference of \eqref{conv_theorem_2} is upper-bounded by \eqref{mapping_T2}. 
Empirical evaluation evidences that the main contribution to the difference arises from the dominant term in \eqref{mapping_T2_1}, and  \eqref{mapping_T2_2}, \eqref{mapping_T2_3}, and \eqref{mapping_T2_4} have zero mean and variance less than $0.01$. On the other hand, \eqref{mapping_T2_1} has a mean of $0.2$ and variance less than $0.01$. This leads to the fact that the sum of \eqref{mapping_T2_2}, \eqref{mapping_T2_3}, \eqref{mapping_T2_4} can be represented as $\epsilon_0$ with $|\epsilon_0| \ll 0.2$. Thus, $\mathbb{F}_{ia}(\boldsymbol{y^{t}})-\mathbb{F}_{ia}(\boldsymbol{z^{t}})$ is asymptotically upper-bounded by $(m_{a}-\bar{m}_{a})(y_{v_{\bar{n}_a+1, a}}^t- z_{w_{\bar{m}_a+1,a}}^t)$ in \eqref{mapping_T2_5}. Similar observation holds for the lower bound.
\end{proof}

\begin{theorem}
There are $\delta \in (0, 1)$ such that 
\begin{align} \label{conv_theorem_goal}
    \|\mathbb{T}(\boldsymbol{y}^t) - \mathbb{T}(\boldsymbol{z}^t)\|_{\infty} \leq \delta \|\boldsymbol{y}^t - \boldsymbol{z}^t\|_{\infty}.
\end{align}
\end{theorem}
\begin{proof}
The $\infty$-norm of \eqref{conv_theorem_2} has the relationship as
\begin{align} \label{conv_theorem_4}
    \|\mathbb{F}(\boldsymbol{y}^t) - \mathbb{F}(\boldsymbol{z}^t)\|_{\infty} &\leq (m_a - \bar{m}_a)\| \boldsymbol{y}^t - \boldsymbol{z}^t \| + \epsilon_0.
\end{align} \label{conv_theorem_5}
Since the mapping $\mathbb{T}$ involves a minimum over $\mathbb{F}$,
    \begin{align} \label{conv_theorem_6}
        \|\mathbb{T}_{ia}(\boldsymbol{y}^t) - \mathbb{T}_{ia}(\boldsymbol{z}^t)\|_{\infty} &= \| \min_{b \in \mathcal{A} \backslash a} \mathbb{F}_{ib}(\boldsymbol{z}^t) -\min_{b \in \mathcal{A} \backslash a} \mathbb{F}_{ib}(\boldsymbol{y}^t) \|_{\infty} \nonumber \\
        & \leq \|\mathbb{F}(\boldsymbol{y}^t) - \mathbb{F}(\boldsymbol{z}^t)\|_{\infty}.
    \end{align}
    Thus, it follows that
    \begin{align} \label{conv_theorem_7}
        \|\mathbb{T}(\boldsymbol{y}^t) - \mathbb{T}(\boldsymbol{z}^t)\|_{\infty} &= \|\mathbb{T}_{ia}(\boldsymbol{y}^t) - \mathbb{T}_{ia}(\boldsymbol{z}^t)\|_{\infty} \nonumber \\
        & \leq (m_a - \bar{m}_a)\| \boldsymbol{y}^t - \boldsymbol{z}^t \|_{\infty} + \epsilon_0 \nonumber \\
        &= \delta_0 \| \boldsymbol{y}^t - \boldsymbol{z}^t \|_{\infty},
    \end{align}
    where $\delta_0$ is given by
    \begin{align} \label{conv_theorem_delta}
        \delta_0 &= (m_a - \bar{m}_a) + \frac{\epsilon_0}{\|\boldsymbol{y}^t - \boldsymbol{z}^t\|_{\infty}},
    \end{align}
    Empirical observation shows that $(m_a - \bar{m}_a) \rightarrow 1$ as $t \rightarrow \infty$. Once the algorithm converges, the connection between vehicle-$i$ and RSU-$a$ does not affect other connections. 
    Furthermore, $m_a$ represents such a state where vehicle-$i$ selects RSU-$a$, while $\bar{m}_a$ represents the state where vehicle-$i$ does not connect to RSU-$a$. 
    Therefore, $m_a-\bar{m}_a \rightarrow 1$, and \eqref{conv_theorem_delta} becomes
    \begin{align} \label{conv_theorem_delta_final}
        \delta_0 &= 1 + \frac{\epsilon_0}{\|\boldsymbol{y}^t - \boldsymbol{z}^t\|_{\infty}} = 1 + \epsilon.
    \end{align}
    where $\epsilon = \frac{\epsilon_0}{\|\boldsymbol{y}^t - \boldsymbol{z}^t\|_{\infty}} \rightarrow 0$ as $t \rightarrow \infty$. For $\delta \in (0, \delta_0)$, \eqref{conv_theorem_goal} holds. 
    Therefore, the resulting algorithm is non-expansive asymptotically with $\epsilon \rightarrow 0$ as $t\rightarrow \infty$.
\end{proof}
\vspace{-3mm}
Hence, under bounded difference assumptions and diminishing perturbation, the message-passing update mapping $\mathbb{T}$ is non-expansive and converges to a fixed point.

In addition to convergence, the proposed algorithm guarantees the optimality under a fixed point solution. 
The global optimality of the resulting solution is established via proof by contradiction. Suppose that the variable assignment associated with the convergent messages is distinct from the optimal solution of \eqref{opt1}.
Let $\Bar{x}$ and $x^{\star}$ denote the solution of the developed algorithm and the optimal solution, respectively. 
The superscript is dropped for a simplified representation of the convergent message. Since the optimal solution has the unique smallest objective value, the corresponding objective values satisfy
\begin{align} \label{opt_diff}
    \sum_{i \in \mathcal{V}} \sum_{a \in \mathcal{A}} \mathcal{T}_{ia}(\Bar{n}_a) \Bar{x}_{ia} > \sum_{i \in \mathcal{V}} \sum_{a \in \mathcal{A}} \mathcal{T}_{ia}(n_a^{\star})x_{ia}^{\star},
\end{align}
where $\Bar{n}_a = \sum_{i\in \mathcal{V}} \Bar{x}_{ia}$ and $n^{\star}_a = \sum_{i\in \mathcal{V}} x_{ia}^{\star}$. Since $\Bar{x}$ and $x^{\star}$ differ, there must be some vehicles such that RSU associations differ between them. The indices of such vehicles are collected to define the following set as
\begin{align} \label{opt_set1}    
        \mathcal{F}=\{i \in \mathcal{V}: \Bar{x}_{ia} \neq x_{ia}^{\star},~a\in \mathcal{A}\}.
\end{align}
For each $i \in F$, define the RSUs that serve vehicle-$i$ as
\begin{align} \label{opt_set}    
        w(i) &= \{ a \in \mathcal{A}: \Bar{x}_{ia} = 1, x_{ia}^{\star} = 0 \}, \nonumber \\
        c(i) &= \{ a \in \mathcal{A}: \Bar{x}_{ia} = 0, x_{ia}^{\star} = 1 \}.    
\end{align}
Note that both RSUs are uniquely identified: RSU $w(i)$ is selected by the proposed algorithm, while $c(i)$ is the one in the optimal configuration. Thus, By construction of the decision metrics $p_{iw(i)}$ and $p_{ic(i)}$ in \eqref{bvalue}, it follows that $p_{iw(i)} < 0$ and $p_{ic(i)} \geq 0$. 

To analyze the objective function gap, consider the total metric difference between $\Bar{x}$ and $x^*$. 
The sum of metrics $p_{ic(i)}$ is evaluated for $i \in F$ as
\begin{align} \label{b_ic_sum}
    \sum_{i \in F} p_{ic(i)} =& \sum_{i \in F} \big(\alpha_{ic(i)}+\rho_{ic(i)}\big)\nonumber \\
    =& \sum_{i \in F} \Big(\mathcal{T}_{ic(i)}(n_{c(i)}^{\star}) + \rho_{ic(i)} \nonumber \\
    &+\sum_{j \in F \backslash \{i\}} \big(\mathcal{T}_{jc(i)}(n_{c(i)}^{\star}) + \rho_{jc(i)}\big) x_{jc(i)}^{\star} \nonumber \\
    &-\sum_{j \in F \backslash \{i\}} \big(\mathcal{T}_{jw(i)}(\Bar{n}_{w(i)}) + \rho_{jw(i)}\big)\Bar{x}_{jw(i)}\Big).
\end{align}
This is reorganized with respect to $\mathcal{T}$ and $\rho$ to obtain
\begin{align}
    \sum_{i \in F} p_{ic(i)} =& \sum_{i \in F} (\mathcal{T}_{ic(i)}(n_{c(i)}^{\star})x^{\star}_{ic(i)} - \mathcal{T}_{iw(i)}(\Bar{n}_{w(i)}) \Bar{x}_{iw(i)}) \nonumber\\
     &+ \sum_{i \in F} (\rho_{ic(i)} - \rho_{iw(i)})\nonumber \\
    \equiv & \Delta T + \Delta \rho,\label{deltarho}
\end{align}
where $\Delta T$ and $\Delta \rho$ are the differences in delay costs and message values expressed, respectively, as
\begin{align}
    \Delta T &= \sum_{(i,a): i \in F} (\mathcal{T}_{ia}(n_a^{\star})x^{\star}_{ia} - \mathcal{T}_{ia}(\Bar{n}_a) \Bar{x}_{ia})\\
    \Delta \rho &= \sum_{i \in F} (\rho_{ic(i)} - \rho_{iw(i)}).
\end{align}
The following lemma is used to investigate the sign of $\Delta T$. 
\begin{lemma} \label{Lemma_1_1}
Given vehicle-$i_0$ and RSU-$a_0$ such that $m_{i_0 a_0} < 0$, it holds that
\begin{align} 
    p_{i_0 a} \geq -p_{i_0 a_0}.
\end{align}
\end{lemma}
\begin{proof}
The metric to determine the value of $x_{i_0 a}$ holds that
\begin{align} \label{Lemma_1_2}
    p_{i_0 a} &= \alpha_{i_0 a} + \rho_{i_0 a} \nonumber \\
    &= \alpha_{i_0 a} - \min_{k \neq a} \alpha_{i_0 k} \nonumber \\
    & \geq \min_{k \neq a_0} \alpha_{i_0 k} - \alpha_{i_0 k_0} = -p_{i_0 a_0}.
\end{align}
The inequality of the last line results from the fact that $\alpha_{i_0 a}\geq \min_{k \neq a_0} \alpha_{i_0 k}$ and $\min_{k \neq a} \alpha_{i_0 k} 
\leq \alpha_{i_0 k_0}$.
\end{proof}
Using the fact that $\rho_{ic(i)} = -\alpha_{iw(i)}$ and Lemma~\ref{Lemma_1_1}, the difference $\rho_{ic(i)} - \rho_{iw(i)}$ is upper-bounded by
\begin{align}
    \rho_{ic(i)} - \rho_{iw(i)} = -\alpha_{iw(i)} -\rho_{iw(i)} = -p_{iw(i)} \leq p_{ic(i)}.
\end{align}
The sum of the corresponding metrics holds that $\sum_{i \in F} p_{ic(i)} \geq \sum_{i \in F} (\rho_{ic(i)} - \rho_{iw(i)}) = \Delta \rho$. However, \eqref{deltarho} leads to $\Delta \rho + \Delta T = \sum_{i \in F} p_{ic(i)} \geq \Delta \rho$, implying$\Delta T \geq 0$. This solution has a smaller objective value than the optimal solution, which contradicts the hypothesis.

These theoretical results can be further validated through the numerical simulation presented in Section \ref{results}.

\section{Numerical Results} \label{results}
\begin{table}
	\centering
	\caption{Simulation setup.}
	\begin{tabular}[t]{lcc}
		\hline
		\textbf{Parameters} & \textbf{Values}\\
		\hline
		Total number of RSUs, $A$ & 7\\
		Total number of vehicles, $V$ & [10-100]\\
		Number of CPUs at RSU-$a$, $k_a$ & [4-8] \\
		Required transmission data, $L_i$ (Mbits) & 1 \\
		Computation workload, $C_i$ (Gcycles) & 1.2 $\pm\,30\%$ \\
		CPU frequency of RSU-$a$, $f_a$ (GHz) & [4-8] \\
		CPU frequency of vehicle-$i$, $f^{\mathrm{loc}}_i$ (GHz) & [2-3] \\
        Maximal allowable processing delay, $t^{\mathrm{max}}_i (s)$ & 0.6 \\
        Allowable Number of Vehicles associated with RSU-$a$, $N_a$& 20\\
		Maximal transmission power, $p_i$ (mW) & 100 \\
		Allocated channel bandwidth to vehicle-$i$, $B_i$ (MHz) & 2 \\
		Average power of noise at RSU-$a$, $\sigma^2$ (pW) & 0.2\\
        Circular coverage radius of RSU-$a$, $d_a$ (m) & 300\\
        Communication bandwidth, $B_p$ (MHz) & 2  \\
        Carrier frequency, $f_{\mathrm{carr}}$ (MHz) & 915\\
        Antenna gain, $A_d$ & 4.11\\
        Path loss exponent, $d_e$ & 2.8\\
        Message iteration count & 15\\
		\hline
	\end{tabular}
	\label{table:parameter}
\end{table}
This section presents numerical evaluations of the proposed algorithm under various VEC configurations. The simulation parameters are summarized in Table~\ref{table:parameter}.
The simulation environment models a $1$-km bidirectional road with six $3.5$-m-wide lanes, where three and four RSUs are placed alternately on both sides, each covering approximately $300$m \cite{fhwa2018ops}.
The wireless channel is modeled by free-space path loss given by $\bar{h}_{ia} = \bar{\zeta}(d_{ia}) = A_d([3 \times 10^8] / [4\pi f_{\mathrm{carr}} d_{ia}])^{d_e}$ \cite{huang2019deep} . 
The channel fading gain is modeled as $h_{ia} = \bar{h}_{ia} \phi$, where the independent channel fading factor $\phi$ is an exponential random variable with parameter $\lambda=1$. 
Each vehicle is tasked with offloading a computation-intensive application that supports hybrid electric vehicle (HEV) controls. 
The corresponding computational demand $C_i$ is modeled as $1.2$ Gcycles per task with 30\% variability \cite{tan2022decentralized}.
Given that up to 100 vehicles may be present in the simulation and a total of seven RSUs are deployed, the maximum number of vehicles that can be associated with each RSU is set to $N_a =20 >\lceil 100/7 \rceil$.

Several baseline algorithms are considered for performance comparison.
Self-organized clustering (SC) \cite{hou2023intelligent} groups vehicles based on proximity to RSUs, forming fewer clusters than the total number of RSUs to reduce the communication delay. 
Each cluster is associated with the nearest RSU. 
To limit queueing delay, the number of tasks per RSU is restricted so that tasks exceeding this limit are executed locally. 
Each vehicle compares $O(A)$ RSUs, resulting in a total complexity of $O(AV) = O(V^2)$.
In the game-based task assignment (GT) \cite{teng2023game, zhang2019task, Jiang2024poten}, vehicles iteratively compete to pick the RSU with the lowest current offloading delay, recalculating delays each round until reaching an equilibrium. 
Since sorting RSUs with their expected delay requires $O(A)$ operations, this yields an overall complexity of $O(AV) = O(V^2)$.
Balanced matching (BM) \cite{sun2023bargain, Li2024two} focuses on load balancing across RSUs. 
Each vehicle is initially assigned to the RSU with the lowest estimated delay. 
Tasks exceeding the per-RSU offload limit are reassigned to alternate RSUs. 
This requires  $\max(A,N_a)$ computations per iteration, resulting in $O(AN_a) = O(V^2)$ in total.
An evolutionary algorithm (EA) \cite{song2021joint, Zhao2024novel} uses a metaheuristic method for the discrete vehicle-to-RSU assignment problem.
A primal-dual algorithm (PD) \cite{feng2021joint, pervez2025efficient, zhong2022potam, Su2024primal} formulates the problem as a two-phase relaxation-based convexified optimization.
The primal phase determines vehicle-RSU associations using real-valued approximations of binary variables. The dual phase updates the predicted number of served vehicles at each RSU based on complementary slackness \cite{boyd2004convex}. The method alternates between vehicles and RSUs with per-iteration complexity of $O(AV) = O(V^2)$, and the complexity scales with the number of iterations.

\begin{figure}
	\begin{center}
		\includegraphics[width=.9\linewidth, height=.7\linewidth]{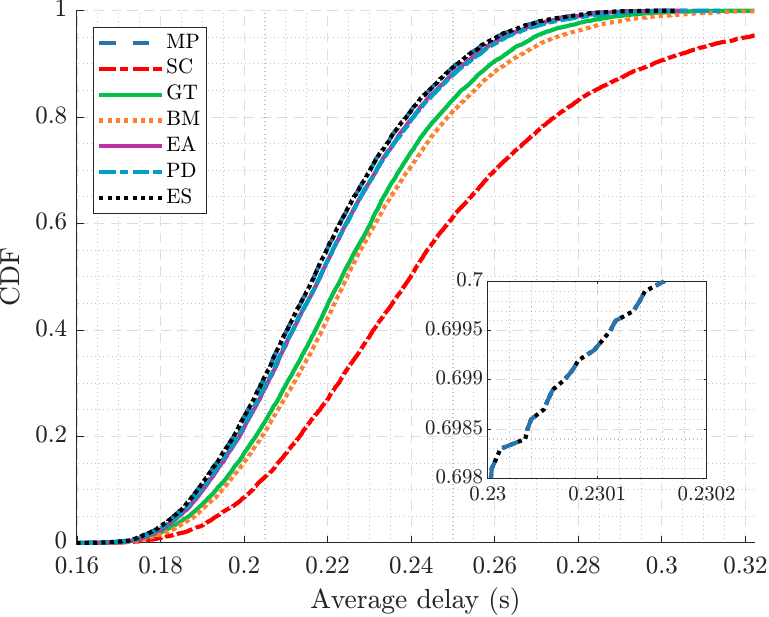}
		\caption{The CDF of the average delay with $V = 10$.} \label{fig:cdf}
	\end{center}
\end{figure}

Fig. \ref{fig:cdf} presents the cumulative distribution function (CDF) of offloading delays across $V = 10$ vehicles.
This shows the statistical distribution of task processing delays across the network.
For benchmarking purposes, the CDF of the global optimum, obtained via exhaustive search (ES), is plotted as well.
The results indicate that the proposed MP algorithm exactly reproduces the delay profile of the ES solution, which confirms that it achieves global optimality, as established in the theoretical analysis of Section \ref{theory}. MP achieves a minimum delay of $0.17$s and a worst-case delay of $0.31$s.
On the other hand, PD and GT methods result in a slightly higher worst-case delay of $0.32$s and show a $5$--$10\%$ degradation in delay performance in the intermediate range of $0.20$s to $0.28$s.
The SC scheme, which relies solely on physical proximity, exhibits a worst-case delay of $0.64$s. 
This indicates that spatial proximity alone is not a valid criterion for delay minimization in heterogeneous computing environments.
The proposed algorithm consistently outperforms existing algorithms across all delay profiles and aligns with the optimal association strategy.



\begin{figure}
	\begin{center}
		\includegraphics[width=.9\linewidth, height=.7\linewidth]{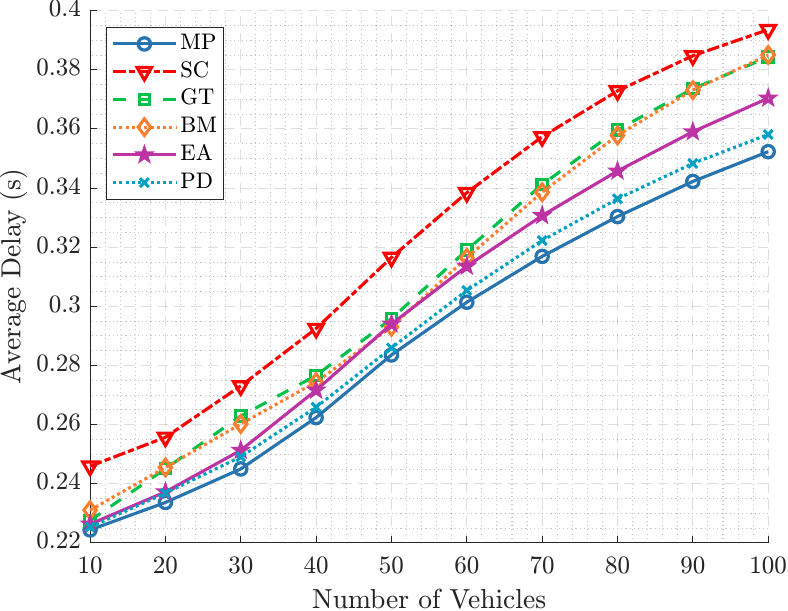}
		\caption{The average delay with respect to the number of vehicles.} \label{fig:veh_num}
	\end{center}
\end{figure}

Fig. \ref{fig:veh_num} illustrates the average processing delay as a function of the number of vehicles in the case of $7$ RSUs, each equipped with $4$--$8$ CPUs. 
The average delay is computed as the mean of all individual vehicle delays, each evaluated according to \eqref{obj_delay2}.
As expected, the overall delay increases with the number of vehicles as the competition heightens for computational resources.
Nevertheless, the proposed algorithm consistently achieves the lowest delay across all traffic densities. 
This is attributed to its ability to accurately estimate the computational load at each RSU and adapt vehicle-to-RSU associations according to CPU availability and channel conditions.
As the vehicle density becomes large, the rate of increase in delay begins to flatten. 
This saturation behavior results from the limited edge computing capacity, beyond which additional offloaded tasks face delays comparable to local execution times. 
Then, MP avoids queue buildup at RSUs by intelligently redirecting certain tasks to be processed locally. 
As a result, in high-load regimes, the average delay converges to the delay of local computing, which highlights a fundamental capacity limit of the edge infrastructure. This observation underscores provisioning sufficient edge computing resources in proportion to vehicular demand.

\begin{figure}
	\begin{center}
		\includegraphics[width=.9\linewidth, height=.7\linewidth]{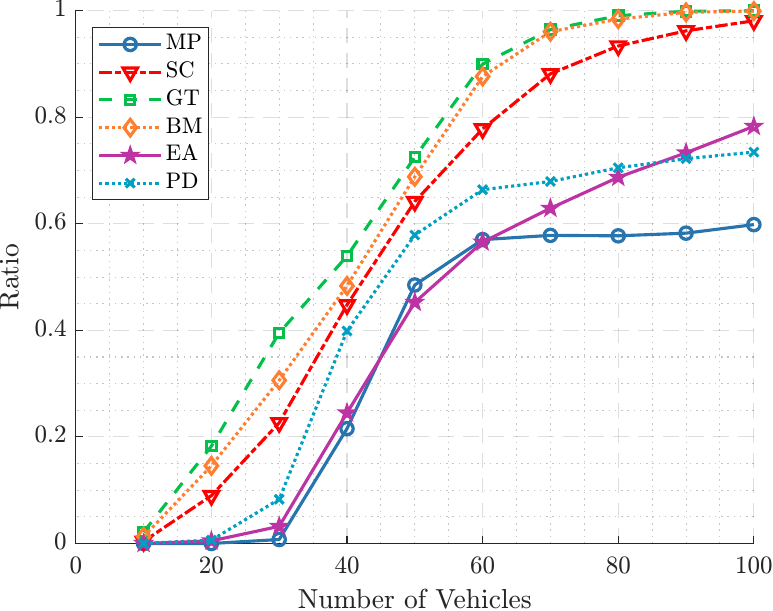}
		\caption{The ratio of tasks exceeding the computation capacities with respect to the number of vehicles.} \label{fig:ratio}
	\end{center}
\end{figure}

Fig. \ref{fig:ratio} shows the proportion of RSUs for which the number of assigned offloaded tasks exceeds the available CPU count.
This ratio is defined as the number of overloaded RSUs, i.e., those with more associated vehicles than CPUs, to the total RSU population.
When the number of vehicles remains below the total CPU count of edge servers, both MP and PD algorithms effectively utilize idle computational resources to avoid task queuing. 
However, as the vehicle count approaches the total CPU capacity, in particular, in the range of $30-60$ vehicles, the task congestion emerges as a critical factor.
In this regime, the MP algorithm significantly outperforms the alternatives by maintaining a lower proportion of overloaded RSUs across all network loads. 
This advantage stems from its combinatorial consideration of the queue length, which explicitly accounts for the maximum number of concurrent tasks that can be handled by each RSU.
While the EA algorithm, another combinatorial method, exhibits a relatively low overload ratio when the number of vehicles is below $60$, its performance deteriorates at higher traffic levels.
This degradation stems from the vast solution space, which makes it infeasible for EA to precisely identify the optimal task allocation.
Furthermore, the relaxation-based PD approach suffers from imprecision, as it approximates binary offloading policies using continuous variables.
Such approximations fail to guarantee the consistency with discrete CPU availability, leading to over-assignment and increased queuing delays.

\begin{figure}
	\begin{center}
		\includegraphics[width=.9\linewidth, height=.7\linewidth]{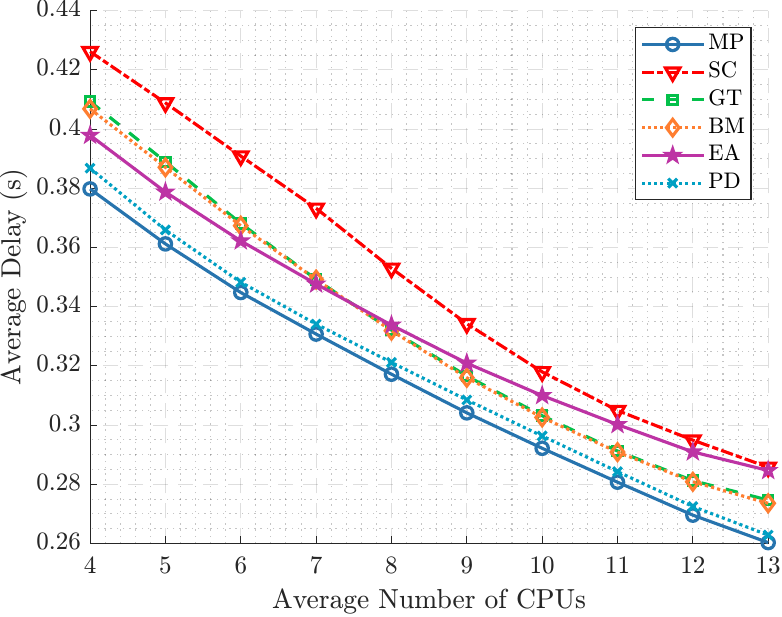}
		\caption{The average delay with respect to the number of CPUs.} \label{fig:cpu_num}
	\end{center}
\end{figure}

Fig. \ref{fig:cpu_num} depicts the average processing delay as a function of the average number of CPUs per RSU for $V = 100$ vehicle cases.
Increasing the number of CPUs reduces the overall delay, since more tasks can be processed concurrently.
While most algorithms benefit from this enhanced computational capacity, the EA algorithm exhibits a noticeably slower decline in delay.
EA tends to retain computationally intensive tasks on local devices in an effort to avoid potential queuing delays at RSUs, even when server resources are abundant. Thus, only light tasks are offloaded, and the increased CPU availability remains under-utilized. 
In contrast, the MP algorithm offloads both light and heavy tasks while maintaining optimal queue configurations across RSUs.
By strategically distributing the workload and considering the full scope of server-side resources, MP achieves a considerable reduction in delay.
This balanced and adaptive allocation exploits CPU availability more effectively than the baseline methods, resulting in greater performance gains as computational resources increase.

\begin{figure}
	\begin{center}
		\includegraphics[width=.9\linewidth, height=.7\linewidth]{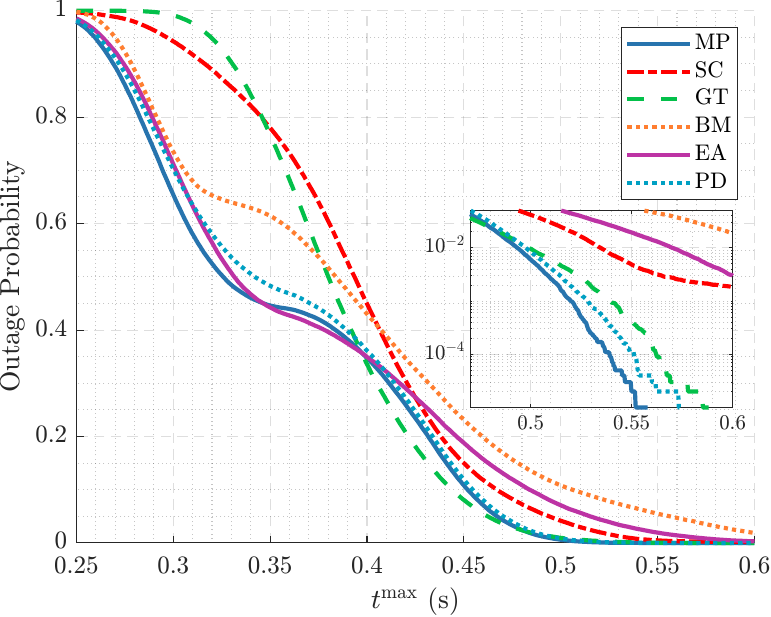}
		\caption{The outage probability.} \label{fig:outage}
	\end{center}
\end{figure}

Fig. \ref{fig:outage} plots the outage probability versus $t_i^{\max}$ for $V = 100$ for a comparative assessment of the reliability of different task offloading strategies. As the deadline becomes more relaxed, all schemes exhibit reduced outage probabilities, as expected. The MP algorithm reaches a plateau between $0.33$s and $0.38$s. 
This behavior arises as it intentionally refrains from further offloading once additional tasks would induce substantial queuing delays, especially when heavier tasks are present. 
Although offloading such tasks can reduce processing time, it also increases the risk of queue overflow. 
Thus, MP maintains a conservative offloading policy unless the delay constraint becomes sufficiently loose to accommodate the resulting queuing risk.
At $t^{\mathrm{max}} = 0.4$s, corresponding to the best-case delay for local computing, the GT algorithm achieves a lower outage than MP by aggressively offloading all tasks to edge servers. While this strategy leverages the high computation power of RSUs, it also leads to significant queuing delays as multiple vehicles converge on the same RSUs, thereby resulting in homogeneous and elevated processing delays.
Also, its lack of queuing awareness increases the likelihood of deadline violations.
The PD algorithm ensures moderate reliability by distributing tasks uniformly through convex optimization. 
However, its use of approximated variables yields imprecise integer assignments on the CPU count, which can result in server overloading and elevated outage rates.
In contrast, MP performs exact integer-based task-to-CPU assignment, achieving perfect matching between the server capacity and the task demand, thereby approaching 100\% accuracy, in comparison to 50\% of the PD scheme.
This precise task-to-CPU matching allows MP to attain high utilization and low queuing congestion, significantly improving reliability.
\begin{figure}
	\begin{center}
		\includegraphics[width=.9\linewidth]{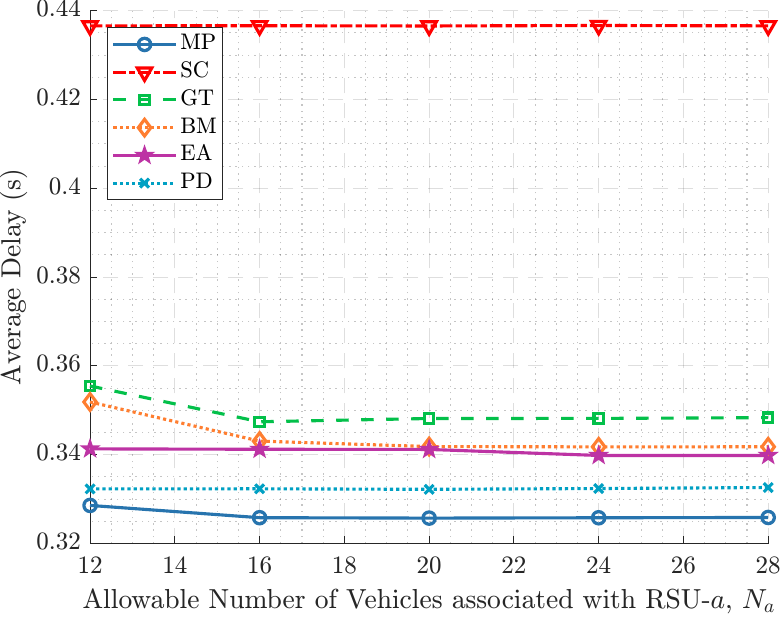}
		\caption{The average delay with respect to the maximum allowed number of vehicles.} \label{fig:max}
	\end{center}
\end{figure}

Fig. \ref{fig:max} depicts the average processing delay as a function of the allowable number of vehicles that can be associated with each RSU for the case of $V=100$ vehicles. 
As the per-RSU capacity increases, the overall delay decreases because a large portion of tasks can be offloaded to edge servers rather than processed locally. However, the benefit becomes marginal beyond $N_a = 16$, as most vehicles are already able to offload their tasks under this configuration. 
The proposed MP scheme consistently exhibits the adaptation to the edge server capacity and maintains the lowest average delay.
\begin{figure}
	\begin{center}
		\includegraphics[width=.9\linewidth]{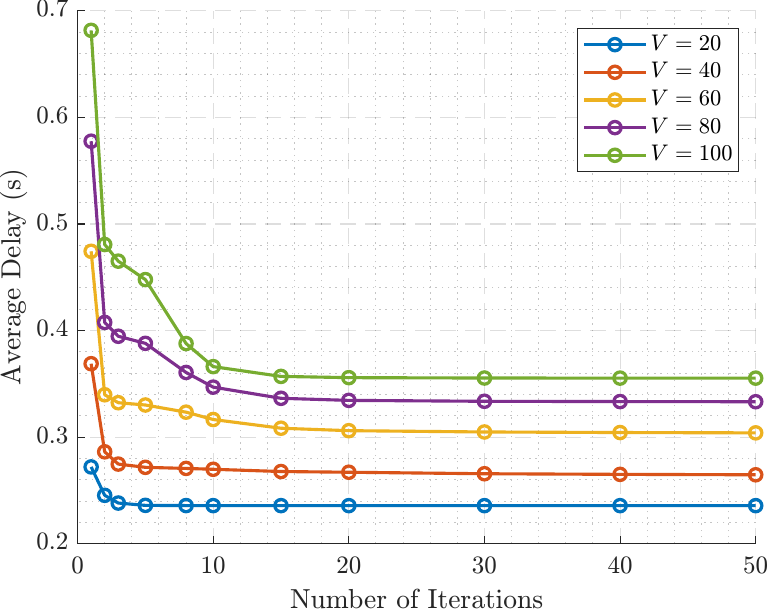}
		\caption{Convergence behavior.} \label{fig:conv}
	\end{center}
\end{figure}

Fig. \ref{fig:conv} shows the convergence behavior of the MP algorithm. The average objective-value trajectories are shown over a maximum of $50$ iterations for $[20,100]$ vehicles under a deployment of $7$ RSUs. The algorithm converges on average within $10$ iterations. 
This consistent convergence confirms the computational efficiency and scalability of the algorithm.

\begin{figure*}
    \centering
    \subfigure[The RSU distribution.]{\includegraphics[width=0.32\textwidth, height = 0.31\textwidth]{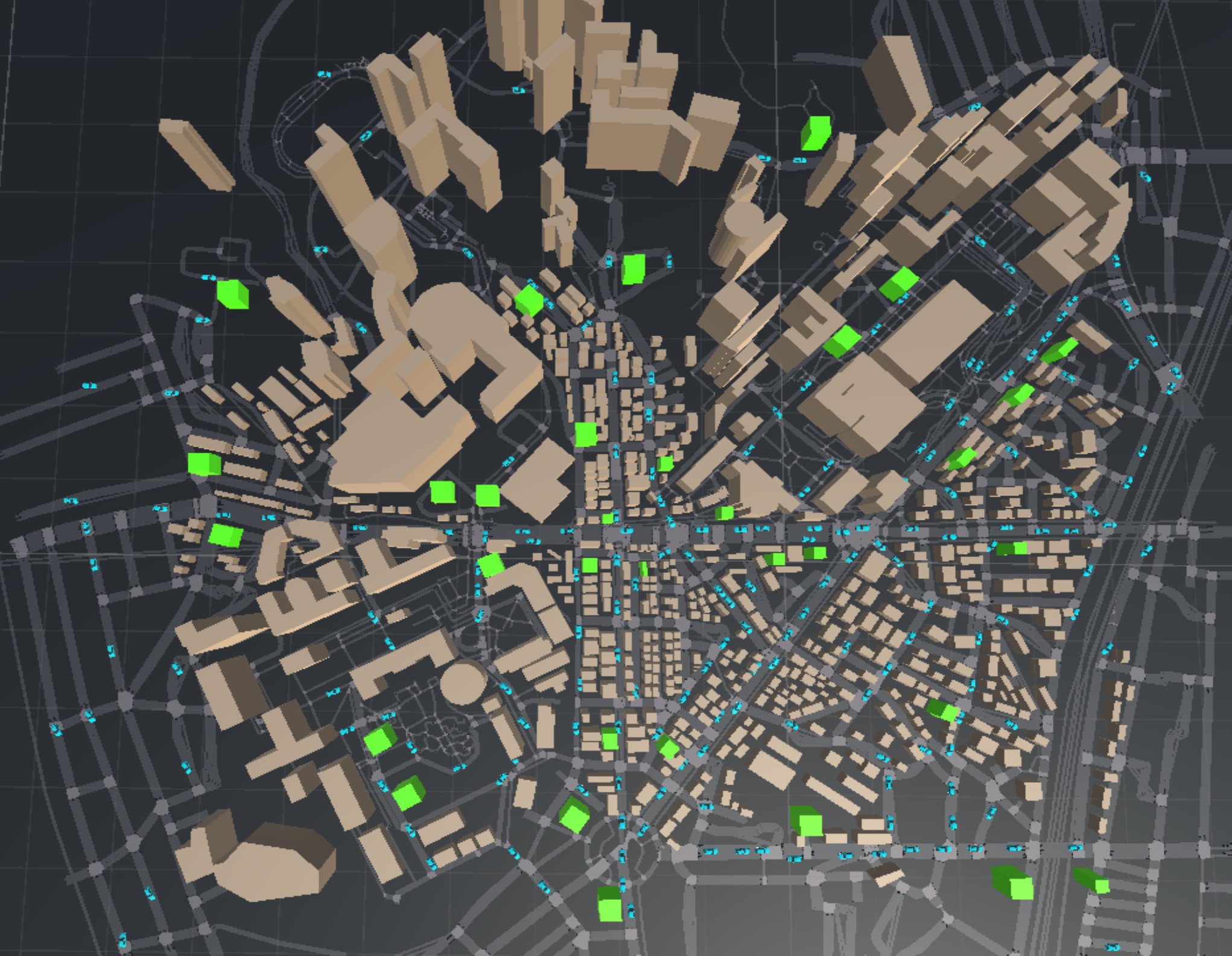}\label{fig_unity_map}}
    \subfigure[The primal-dual algorithm \cite{feng2021joint, zhong2022potam, pervez2025efficient}.]{\includegraphics[width=0.32\textwidth, height = 0.31\textwidth]{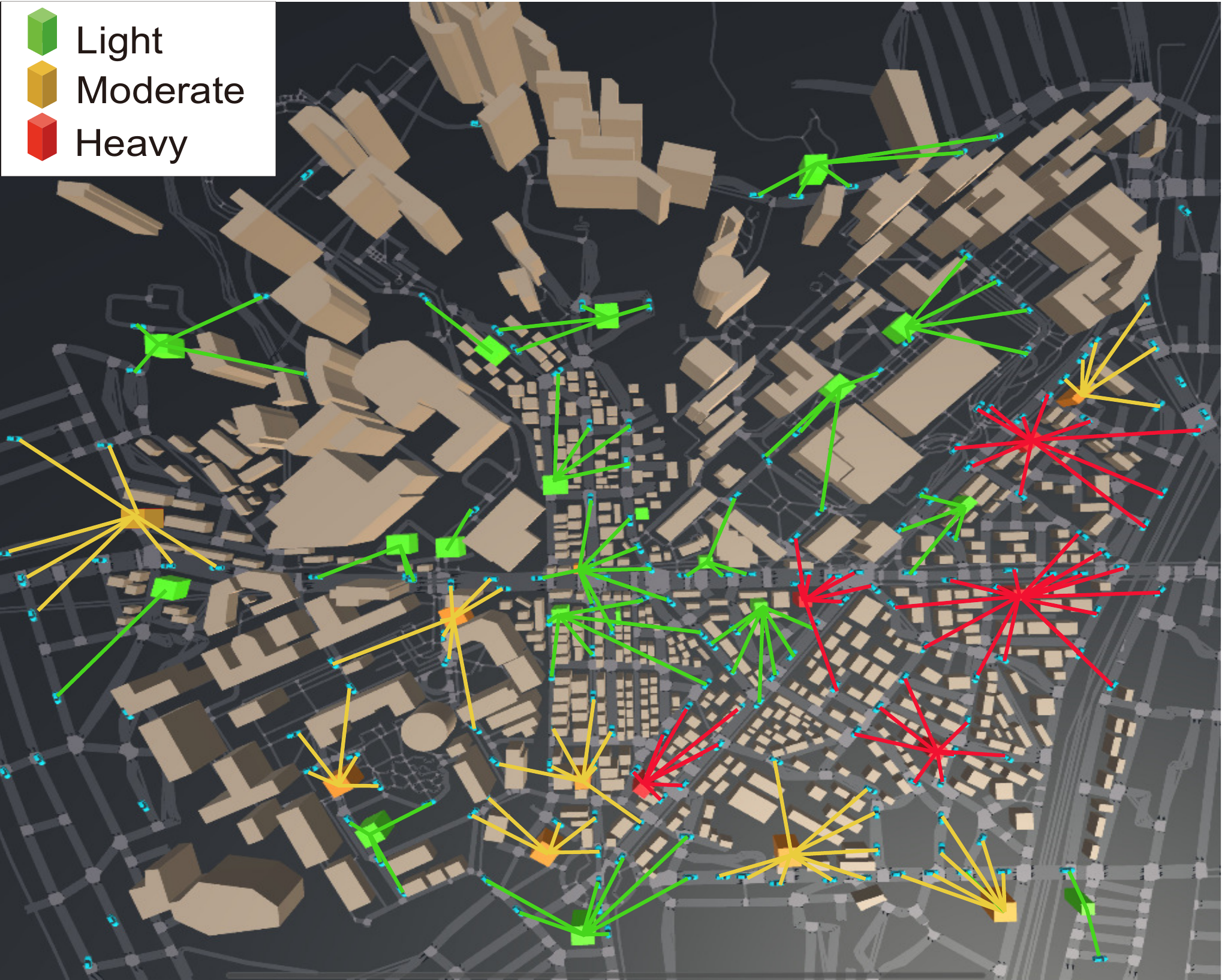}\label{fig_unity_pd}}
    \subfigure[The proposed MP algorithm.]{\includegraphics[width=0.32\textwidth, height = 0.31\textwidth]{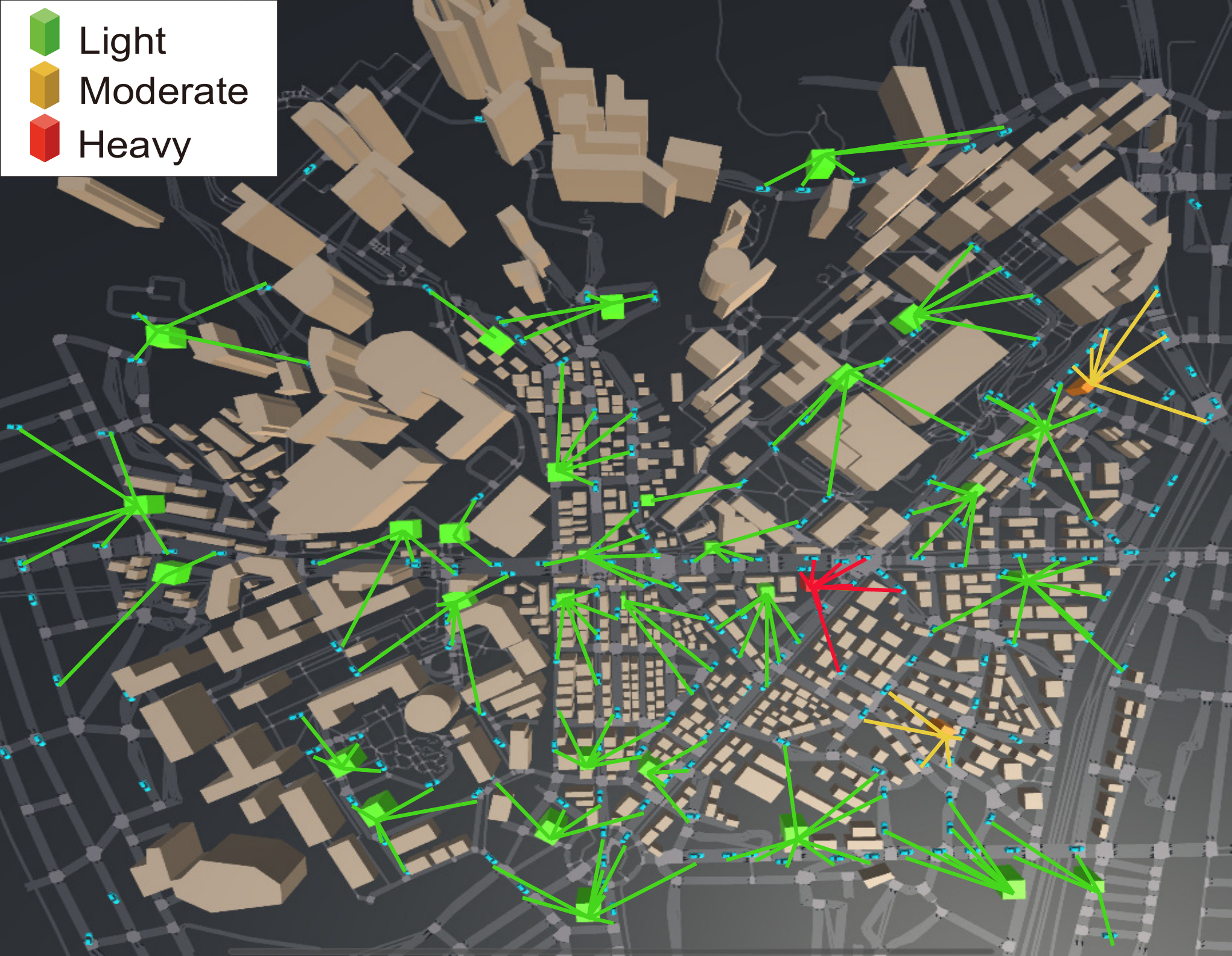}\label{fig_unity_mp}}
    \caption{The verification of feasibility over a real-map digital-twin platform.}
    \label{fig:unity}
\end{figure*}
Fig. \ref{fig:unity} illustrates the feasibility evaluation results obtained using a real-map digital twin platform. The testbed is implemented in Unity, based on the area surrounding Korea University in Seoul with real cell information. 
A total of $33$ sites equipped with $5$G base stations, are deployed as RSUs to support edge computing for $200$ vehicles. The proposed MP algorithm is compared against the PD algorithm in this environment. To evaluate system load balancing, the RSUs are classified into three categories based on their load status: lightly loaded (fewer associated vehicles than CPUs), moderately loaded (1.5 times the CPU count), and busy (more than 1.5 times the CPU count). 
The proposed algorithm obtains significant average delay reductions by redistributing tasks from five overloaded RSUs to underutilized ones.
This improvement is attributed to the precise estimation of $n_a$. Unlike PD approximating $n_a$ with continuous values, the MP algorithm explicitly computes integer values and adjusts vehicle associations accordingly to balance processing delays across RSUs. 
With the exact estimation of $n_a$, the proposed technique determines efficient load-balancing over servers to minimize the overall delay.

These extensive simulation results confirm that the proposed distributed algorithm consistently outperforms existing methods in both performance and convergence efficiency. 
The theoretical guarantees established earlier in this paper are thus substantiated through both extensive numerical experiments and real-world digital twin simulations, demonstrating viability for practical deployment in VEC networks.

\section{Conclusion}
\label{conclusion}
This paper presents a distributed offloading strategy aimed at minimizing network-wide computation latency in VEC systems. 
The proposed method, built upon an MP framework, determines vehicle-to-RSU associations by considering load-dependent processing capabilities of edge servers. 
A theoretical analysis establishes both the convergence and optimality of the proposed scheme.
Simulation results underscore the importance of accurately estimating the number of tasks assigned to each edge server.
In particular, accurately estimating the processing capacity of edge servers can balance resource under-utilization and load congestion.
The proposed technique precisely captures these factors using discrete queue variables, thereby achieving an optimal solution.
A notable limitation also remains. While the proposed approach minimizes the average processing delay, it does not explicitly optimize for worst-case latency, which is a critical factor in outage performance and is heavily influenced by task-priority scheduling.
Therefore, reducing tail latency through dynamic or priority-aware scheduling mechanisms remains a key direction for future work.



\bibliographystyle{IEEEtran}
\bibliography{ref}

\vskip -2\baselineskip plus -1fil
\begin{IEEEbiographynophoto}
{Sungho Cho} received the B.S. degree from Korea University, Seoul, South Korea, in 2024.
He is currently pursuing the M.S. degree with the Department of Electrical and Computer Engineering, University of California, Los angeles. His research interests include vehicle edge computing systems and full-duplex systems optimization.
\end{IEEEbiographynophoto}
\vskip -2\baselineskip plus -1fil

\begin{IEEEbiographynophoto}
{Sung Il Choi} received the B.S. degree in electrical engineering from Korea University, Seoul, South Korea, in 2021. He is currently pursuing the Ph.D. degree with the Department of Electrical Engineering, Korea University. His research interests include communication theory applied to Internet of Things (IoT) and vehicular networks.
\end{IEEEbiographynophoto}

\vskip -2\baselineskip plus -1fil
\begin{IEEEbiographynophoto}
{Seung Hyun Oh} received the B.S. degree from Sejong University, Seoul, Korea in 2023. He is currently pursuing a Ph.D. degree at the School of Electrical Engineering, Korea University, Seoul, Korea. His research interests include learning, communication systems, optimization and virtual testbed platform development.
\end{IEEEbiographynophoto}

\vskip -2\baselineskip plus -1fil

\begin{IEEEbiographynophoto}
{Ian P. Roberts} received the B.S. degree in electrical engineering from Missouri University of Science and Technology and the M.S. and
Ph.D. degrees in electrical and computer engineering from The University of Texas at Austin, where he was a National Science Foundation Graduate
Research Fellow with the Wireless Networking and Communications Group. He is currently an Assistant Professor of electrical and computer engineering with UCLA. He has industry experience developing and prototyping wireless technologies at AT\&T Laboratories, Amazon, GenXComm (startup), and Sandia National Laboratories. His research interests include the theory and implementation of millimeter wave systems, in-band full-duplex, and other next-generation technologies for wireless communication and sensing. In 2023, he received the Andrea Goldsmith Young Scholars Award from the Communication Theory Technical Committee of the IEEE Communications Society.
\end{IEEEbiographynophoto}

\vskip -2\baselineskip plus -1fil
\begin{IEEEbiographynophoto}
{Sang Hyun Lee} 
received the B.S. and M.S. degrees from Korea Advanced Institute of Science and Technology (KAIST) in 1999 and 2001, respectively, and his Ph.D. degree from the University of Texas at Austin in 2011. Since 2017, he has been with the School of Electrical Engineering, Korea University, Seoul, Korea. His research interests include learning, inference, optimization and their applications to wireless communications.
\end{IEEEbiographynophoto}

\end{document}